\documentclass[lettersize,journal]{IEEEtran}
\usepackage{amsmath,amsfonts}
\usepackage{algorithmic}
\usepackage{algorithm}
\usepackage{array}
\usepackage{textcomp}
\usepackage{stfloats}
\usepackage{url}
\usepackage{verbatim}
\usepackage{graphicx}
\usepackage{cite}
\usepackage{amsmath,amsthm}   
\usepackage{amssymb}
\usepackage{nomencl}
\usepackage{multirow}

\usepackage{bm}
\usepackage{lipsum}
\usepackage{makeidx}
\usepackage{enumerate}
\usepackage{color}

\usepackage[font=small,labelfont=bf]{caption}
\usepackage[font=small,labelfont=bf]{subcaption}

\usepackage{tikz,pgfplots}
\usetikzlibrary{calc}
\usetikzlibrary{intersections}
\usetikzlibrary{arrows,shapes}
\usetikzlibrary{positioning}

\newtheorem{corollary}{Corollary}

\newtheorem{lemma}{Lemma}

\theoremstyle{definition}

\allowdisplaybreaks

\begin{document}
\title{NOMA-based Improper Signaling for Multicell MISO RIS-assisted Broadcast Channels}
\author{Mohammad Soleymani$^*$ \emph{Member, IEEE}, 
Ignacio Santamaria$^\dag$ \emph{Senior Member, IEEE},\\
Eduard Jorswieck$^\ddag$ \emph{Fellow, IEEE}, and
Sepehr Rezvani$^\ddag$ \emph{Student Member, IEEE}
 \\ \thanks{ 
$^*$Mohammad Soleymani is with the Signal and System Theory Group, Universit\"at Paderborn, , 33100 Paderborn, Germany   
(e-mail: \protect\url{mohammad.soleymani@sst.upb.de}).  

$^\dag$Ignacio Santamaria is with the Department of Communications Engineering, University of Cantabria, 39005 Santander, Spain (e-mail: \protect\url{i.santamaria@unican.es}).

$^\ddag$ Eduard Jorswieck and Sepehr Rezvani are with Institute for Communications Technology, Technische Universit\"at Braunschweig, 38106 Braunschweig, Germany
(e-mails: \protect\url{{jorswieck,rezvani}@ifn.ing.tu-bs.de})

The work of I. Santamaria has been partly  supported by the project ADELE PID2019-104958RB-C43, funded by MCIN/ AEI
/10.13039/501100011033.

The work of Eduard Jorswieck was supported in part by the Federal Ministry of Education and Research (BMBF, Germany) in the program of ``Souver\"an. Digital. Vernetzt.'' joint project 6G-RIC, project identification number: 16KISK020K and 16KISK031.
}}
\maketitle

\begin{abstract}
In this paper, we study the performance of reconfigurable intelligent surfaces (RISs) in a multicell broadcast channel (BC) that employs improper Gaussian signaling (IGS) jointly with non-orthogonal multiple access (NOMA) to optimize either the minimum-weighted rate or the energy efficiency (EE) of the network. We show that although the RIS can significantly improve the system performance, it cannot mitigate interference completely, so we have to employ other  interference-management techniques to further improve performance. We show that the proposed NOMA-based IGS scheme can substantially outperform proper Gaussian signaling (PGS) and IGS schemes that treat interference as noise (TIN) in particular when the number of users per cell is larger than the number of base station (BS) antennas (referred to as overloaded networks). In other words, IGS and NOMA complement to each other as interference management techniques in multicell RIS-assisted BCs. Furthermore, we consider three different feasibility sets for the RIS components showing that even a RIS with a small number of elements provides considerable gains for all the feasibility sets. 
\end{abstract} 
\section{Introduction}
Reconfigurable intelligent surfaces (RISs) are promising technologies for next generations of wireless communication systems  \cite{huang2020holographic, wu2021intelligent, di2020smart, elmossallamy2020reconfigurable}.  RIS can modulate the channel coefficients, which provides some additional degrees of freedom in the system design. By optimizing the propagation channel, RIS can improve different performance metrics and/or can manage interference.
In this work, we study the performance of RIS in a multicell broadcast channel (BC), which is an interference-limited system. It is known that  RIS can improve the performance of such a system, however, it should be further investigated what is the main role of RIS in such improvements. For instance, are the RIS benefits primarily because of its ability to cancel interference in overloaded systems, or are they due to a coverage improvement? In this paper, we provide an answer for these important questions. 

\subsection{Related work}
Interference has always been among the main performance restrictions  for modern wireless communication systems, and therefore interference-management techniques play an essential role in improving the performance of interference-limited networks \cite{andrews2014will}. 
Treating interference as noise (TIN) is the optimal decoding strategy when the interference level is low \cite{annapureddy2009gaussian}. Additionally, in the presence of strong interference,  the optimal decoding strategy is successive interference cancellation (SIC) in which the signal of the interfering user is decoded at first, and then subtracted from the received signal \cite{sato1981capacity}. 
However, for operational points between these two extreme cases the optimal strategy to handle interference is either  unknown or involves very high complexities. In these cases, we usually have to employ effective, but not necessarily optimal, interference-management techniques such as improper Gaussian signaling (IGS) and/or non-orthogonal-multiple-access (NOMA)-based techniques. 

IGS was employed as an interference-management technique for the first time in \cite{cadambe2010interference}, where the authors showed that IGS can increase the degrees-of-freedom (DoF) of the 3-user IC. In the last decade, IGS has been applied to various interference-limited systems, showing that it can increase the spectral and energy efficiency of such systems with TIN  \cite{javed2018improper, soleymani2020improper, gaafar2017underlay, amin2017overlay, soleymani2019energy, lameiro2015benefits, lameiro2017rate, lagen2016superiority, lagen2016coexisting,  zeng2013transmit, ho2012improper, soleymani2019improper,  soleymani2019robust, soleymani2019ergodic, Sole1909:Energy, tuan2019non, nasir2020signal, nasir2019improper, nguyen2021improper, yu2020improper, yu2020improper2, park2013sinr, yu2021maximizing}. For instance, the papers \cite{gaafar2017underlay, amin2017overlay, soleymani2019energy, lameiro2015benefits} showed that IGS can improve the rate and/or energy efficiency (EE) of the secondary user in cognitive radio systems. In \cite{lagen2016coexisting, zeng2013transmit, soleymani2019improper, ho2012improper}, it is shown that IGS can enlarge the rate region of the 2-user IC. 
The papers \cite{soleymani2019robust,soleymani2019ergodic} showed that IGS can increase the spectral efficiency (SE) of the 2-user IC under realistic assumptions regarding the channel state information (CSI).  These papers present a few representative examples of the superiority of IGS in interference-limited systems. For a more detailed overview of the applications of improper signaling in wireless communications,  the reader is referred to \cite{javed2020journey}.

TIN is a simple and practical decoding scheme, which is shown to be optimal from a generalized-DoF (GDoF) perspective when the interference level is low \cite{geng2015optimality}. However, TIN can be highly suboptimal in the presence of strong interference.  Under strong interference, the optimal decoding strategy is SIC, which is the basic technique in NOMA. As an example, \cite{ding2017survey} has shown that under strong interference NOMA can significantly improve the performance of TIN in a single-input single-output (SISO) BC. Of course, applying SIC among all users entails some challenges especially in large-scale multicell BCs with multiple-antenna BSs  and many users per cell \cite{clerckx2021noma}.  One of the challenges to implement SIC is to obtain the optimal user ordering, which may require to solve an  NP-hard problem \cite{ni2021resource}. Moreover, applying SIC with many users per cell requires complex devices to be able to decode and subtract the signal of several users. To address these issues, hybrid schemes have been proposed in which users are clustered so that SIC is only applied to the users in a  given cluster.

How to solve the aforementioned challenges remains an open question in RIS-assisted systems.
RIS is an emerging technology, which has been shown to be able to improve the spectral and energy efficiency of various systems  \cite{huang2020holographic, wu2021intelligent, di2020smart, elmossallamy2020reconfigurable,  huang2019reconfigurable, wu2019intelligent, kammoun2020asymptotic, pan2020multicell, zhang2020intelligent, yang2020risofdm, zuo2020resource, mu2020exploiting, ni2021resource, yang2021reconfigurable, yu2020joint, yu2021maximizing, zhou2020framework, shaikh2022downlink, soleymani2023spectral}. 
For example, the papers \cite{huang2019reconfigurable, wu2019intelligent, kammoun2020asymptotic, yu2020joint} considered various performance metrics such as global energy efficiency (GEE), sum rate, minimum SINR and showed that RIS can improve the spectral and energy efficiency of a single-cell BC with a multi-antenna BS and multiple single-antenna users. The paper \cite{pan2020multicell} considered a multiple-input, multiple-output (MIMO) multicell BC and showed that RIS can improve the weighted-sum rate of the system. 
In \cite{zhang2020intelligent}, it was shown that RIS can improve the weighted-sum rate for secondary users in an underlay cognitive radio (CR) system. In addition, RIS can increase the achievable rate of an  orthogonal-frequency-division-multiplexing (OFDM) scheme for a single-antenna point-to-point RIS-assisted system, as shown in \cite{yang2020risofdm}. The authors in \cite{zhou2020framework} considered instantaneous, but noisy channel state information (CSI) at transmitters and proposed a  framework for robust designs in MISO RIS-assisted systems.

Despite these promising results, there are still many research questions regarding the use of RIS in multicell BCs that must be thoroughly investigated. For example, RIS can modulate channels and cancel leakage and/or interference links \cite{jiang2022interference}. It should be clarified whether and/or how we can exploit this feature in different interference-limited systems. For instance, in a single-cell SISO BC with two users, both the useful and interfering signals are received via the same channel, and canceling the interference link cancels the useful signal as well \cite{ho2012instantaneous}. This means that there are some limitations on the use of RIS as an interference-management technique.  As a result, other interference-management techniques may be needed to fully exploit benefits of RIS, especially in overloaded networks. Indeed, as we will show in this work, RIS cannot fully handle interference in overloaded multicell BCs, but it can greatly improve the system coverage and even increase the benefits of IGS and NOMA. Another important parameter for the performance of RIS is its position in the system. Due to a more severe large-scale path loss of RIS links, it is important to have a line-of-sight (LoS) link between BS and RIS as well as between RIS and users. Thus, RISs should be located close to users to reduce the path loss and increase the probability of having a LoS link. Otherwise, the benefit of RIS would be minor \cite{soleymani2022improper}.

\subsection{Motivation}
\begin{table*}
\centering
\footnotesize
\caption{A brief comparison of the most related works.}\label{table-1}
\begin{tabular}{|l|c|c|c|c|c|c|c|c|c|}
	\hline
 & This paper& \cite{soleymani2022improper}&\cite{nasir2020signal,tuan2019non}& \cite{yu2020joint, yu2021maximizing}&
\cite{javed2018improper, soleymani2020improper, gaafar2017underlay, amin2017overlay,  lameiro2015benefits, lameiro2017rate, lagen2016superiority, lagen2016coexisting,  zeng2013transmit, ho2012improper, soleymani2019improper,  soleymani2019robust, soleymani2019ergodic}&\cite{   kammoun2020asymptotic, pan2020multicell, zhang2020intelligent, yang2020risofdm}&
\cite{mu2020exploiting, zuo2020resource,yang2021reconfigurable}&
\cite{huang2019reconfigurable}\\
\hline
IGS&$\surd$&$\surd$&$\surd$&$\surd$&$\surd$&&&\\
\hline
RIS&$\surd$&$\surd$&&$\surd$&&$\surd$&$\surd$&$\surd$\\
\hline
NOMA&$\surd$&&$\surd$&&&&$\surd$&\\
\hline
EE metrics&$\surd$&$\surd$&&&&&&$\surd$\\
\hline
		\end{tabular}
\normalsize
\end{table*} 
In our previous study \cite{soleymani2022improper}, we 
showed that IGS can improve the spectral and EE of multicell RIS-assisted BCs with TIN, so IGS and RIS provide additive benefits. As indicated, TIN is optimal only when the interference level is low. However, it is not the case in modern wireless communication systems, where the performance of such systems are mainly limited by interference \cite{andrews2014will}. Thus, it is important to investigate whether and how IGS  can improve the performance of RIS-assisted systems  when it is combined with another advanced interference-management technique, which employs SIC to cancel part of interference.This, indeed, motivates us to investigate the performance of more advanced interference-management techniques (a combination of IGS and NOMA) in a multicell RIS-assisted BC.

IGS has been shown to be a powerful interference-management technique; however, it also has its own limitations. 
Particularly, when we increase the number of resources in frequency (by OFDM \cite{soleymani2018improper}), in time (by time sharing \cite{hellings2018improper}), or in space (by MIMO \cite{soleymani2020improper}) for a fixed number  of users, the benefits of IGS decrease or may even vanish. 
Thus,  one would expect that IGS might not be beneficial if NOMA is applied since a significant part of the interference would be canceled by SIC. However, we show that IGS can provide a considerable gain with NOMA. Indeed, our results show that NOMA and IGS can be employed as complementary and mutually beneficial interference-management techniques. Note that \cite{nasir2020signal, tuan2019non} also showed that IGS with NOMA can improve the minimum rate of users in a multicell BC without RIS. In this paper, to the best of our knowledge, we apply a NOMA-based IGS scheme to RIS-assisted systems to improve rate and EE metrics for the first time in the literature.  Additionally, we will provide  a comprehensive analysis on the RIS performance in practical scenarios with realistic assumptions and try to answer the following research questions:
\begin{itemize}
\item When do we need interference-management techniques such as NOMA and IGS  in RIS-assisted systems? 

\item How many RIS components do we need and how should they be configured/optimized to reap benefits from RIS?
\end{itemize}
These questions are very important for implementing RIS in practice since we should avoid adding unnecessary complexities introduced by considering additional interference-management techniques and/or adding more RIS components. We provide a brief comparison between our work and most related papers in the literature in Table \ref{table-1}.

\subsection{Contribution}
In this paper, we propose NOMA-based IGS schemes to improve the spectral and energy efficiency of multicell RIS-assisted BCs. We make realistic assumptions regarding the channel model, large-scale path loss, and reflecting coefficients of RIS components by considering different feasibility sets for the reflecting coefficients based on the models in \cite{wu2021intelligent}. 
We consider a multicell BC in which multiple-antenna BSs serve several single-antenna users.  
Furthermore, we consider multiple RISs in the system such that each cell has at least one RIS. The reason is that, as we showed in \cite{soleymani2022improper}, the position of RISs plays a key role, and distributed implementation of RIS locations can considerably outperform a co-located RIS implementation. 

 We aim at maximizing either the minimum-weighted rate or the EE of users, which are practical metrics, widely used in the literature \cite{zeng2013transmit, Sole1909:Energy, tuan2019non, nasir2020signal, nasir2019improper, yu2020improper, soleymani2021distributed, soleymani2020rate}. Employing a rate/EE profile technique, the whole rate/EE region can be characterized by solving the maximization of the minimum-weighted rate/EE \cite{zeng2013transmit, Sole1909:Energy}. We show that these metrics can be improved even by RIS with a small number of elements.
However, a RIS alone is not able to fully cancel intracell interference when BSs are overloaded, which enforces us to employ  additionally other interference-management techniques such as IGS and NOMA. In the proposed NOMA-based IGS schemes, we divide users into several clusters with two users each. We then apply NOMA to each 2-user cluster and employ IGS to handle the remaining interference.
The proposed NOMA-based IGS scheme substantially outperforms the IGS schemes in \cite{soleymani2022improper} as well as NOMA with PGS from both spectral and energy efficiency points of view.

Three different feasibility sets for the RIS components  are considered. 
In the first set, whose performance can be viewed as upper bound of more practical RIS-assisted systems, the amplitude and the phase of each RIS component can be optimized independently. Our proposed algorithms converge to a stationary point of the considered optimization problems for this feasibility set. In the second, a more practical feasibility set is considered in which only the phase of each RIS component can be optimized, while the amplitude is fixed. This feasibility set can be referred to as  the reflecting RIS \cite{di2020smart}. 
In the third, we employ the model in \cite{abeywickrama2020intelligent}  in which the amplitude is not constant, but it is a function of the phase of RIS components. Our numerical results suggest that RIS can substantially improve the performance of the system with all these feasibility sets even with a small number of RIS components. 
Furthermore, reflecting RIS can perform close to the upper bound in some regimes. However, the performance gap increases with the transmission power of BSs.

We, furthermore, consider the impact of various parameters such as the number of BS antennas, the number of RIS components, and the number of users per cell on the performance of our proposed interference-management techniques. 
The goal is to find operational points that improve the system performance. 
Our numerical results show that the proposed NOMA-based IGS scheme can provide a significant gain when BSs are overloaded, i.e., the number of users per cell is larger than the number of transmit antennas at BSs.  
The benefits of the proposed schemes increase with the number of users per cell for a fixed number of BS antennas. However, the improvements are diminishing in the number of BSs antennas.
Indeed, the more overloaded the BSs are, the more gains NOMA-based IGS provides. This finding is also in line with our previous studies in \cite{soleymani2020improper, soleymani2022improper}; which altogether corroborate that the higher the interference level is, the more benefits IGS provides.

The main contributions of this work can be summarized as in the following:
\begin{itemize}
\item We propose NOMA-based IGS schemes for multicell RIS-assisted BCs and show that RIS cannot fully mitigate interference even when the RIS components are properly optimized.  Therefore, RIS has to be employed along with a suitable interference-management technique  (such as IGS/NOMA) to fully reap its  benefit. 

\item We show that the combination of NOMA with IGS techniques can substantially improve the performance of RIS-assisted multicell downlink channels whenever the BSs are overloaded. 

\item Our results show that RIS can provide significant gains even with a relatively small number of RIS components, suggesting that reconfigurable intelligent surfaces can be a promising and practical technology.
\end{itemize}

\subsection{Paper outline}
This paper is organized as follows. Section \ref{sec-sym} presents the RIS and system model, states the assumptions for the proposed NOMA-based IGS scheme, and presents the rate/EE metrics.  
Section \ref{sec-ii} and Section \ref{EE-sec}, respectively, present the proposed algorithms to solve the minimum-weighted rate and EE maximization  problems. 
Section \ref{sec-v-nr} provides numerical examples, and Section \ref{sec-con} concludes the paper.
\section{System model}\label{sec-sym}

\begin{figure}[t!]
    \centering
\includegraphics[width=.34\textwidth]{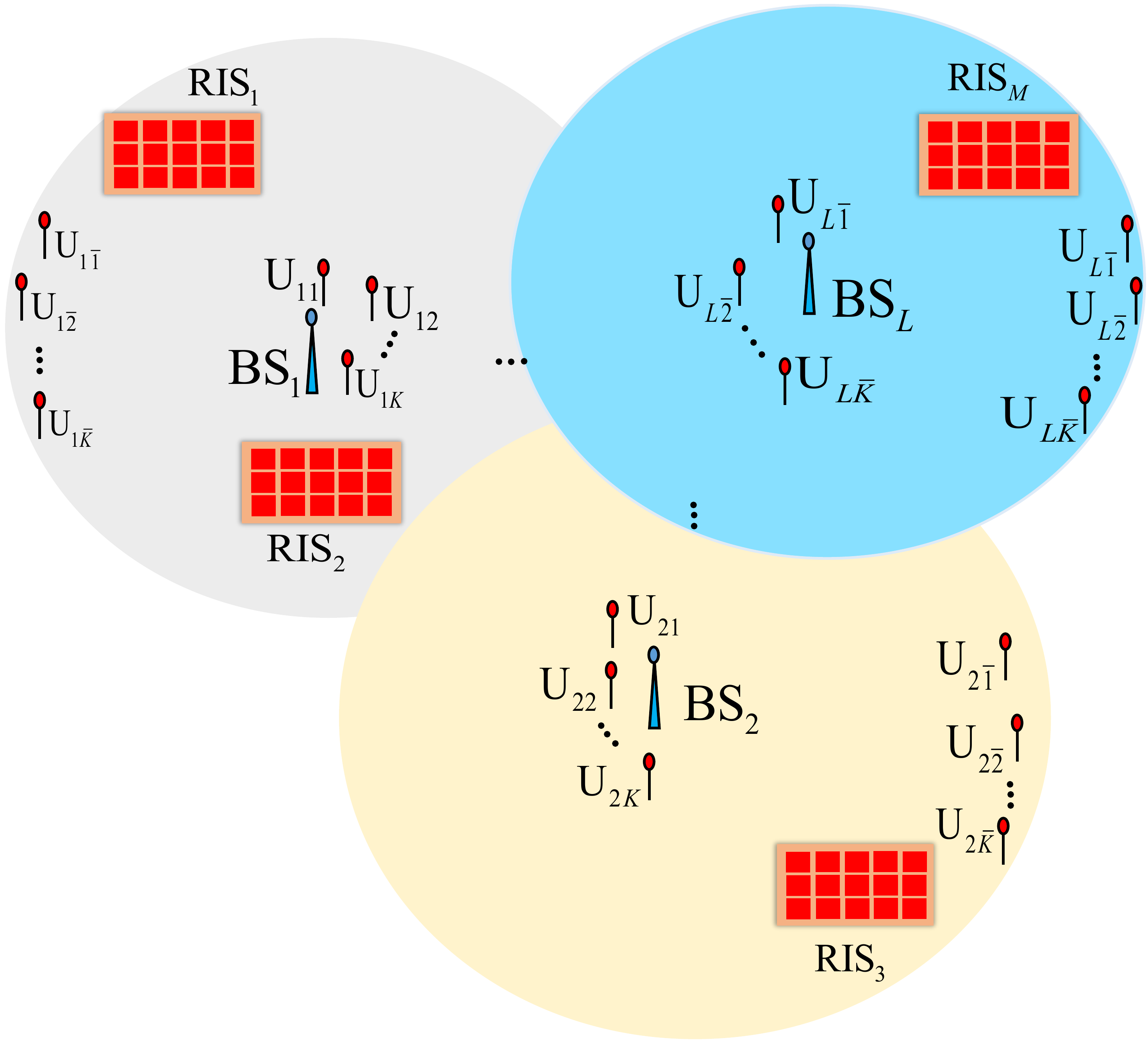}
     \caption{A multicell broadcast channel with RIS.}
	\label{Fig-sys-model}
\end{figure}
We consider a multicell broadcast channel with $L$ multi-antenna base stations (BSs) in which each BS has $N_{BS}$ antennas and serves $2K$ single-antenna users. We assume that there are $M\geq L$ RISs with $N_{RIS}$ reflecting elements in the system to assist the BSs. 
In \cite{soleymani2022improper}, it was shown that a distributed implementation of RIS can outperform a centralized RIS in which all RIS elements are co-located. The reason is that the link through RIS may suffer from a more severe large-scale path loss in comparison with the direct link. Thus, we assume that in each cell there is at least one RIS. Note that we consider symmetric BSs/cells/RISs for notational simplicity, but the analysis can be easily extended to an asymmetric scenario with different number of users associated  to each BS, different number of BS antennas, or different number of reflecting elements  in each RIS.

\subsection{RIS model}\label{sec-ris}
In this paper, we  assume perfect and global CSI similar to \cite{ni2021resource, huang2019reconfigurable, wu2019intelligent, kammoun2020asymptotic, pan2020multicell, zhang2020intelligent, zuo2020resource, mu2020exploiting, yang2021reconfigurable, yu2020joint, jiang2022interference}. We employ the channel model  in \cite{pan2020multicell, soleymani2022improper} for RIS-assisted systems. For the sake of completeness, we review the model in this subsection and briefly discuss its features. In RIS-assisted systems, there can be two types of links between a transmitter and a receiver: a direct link and links through RISs. Thus, the channel between BS $i$ and user $k$ associated to BS $l$, i.e., u$_{lk}$, is \cite{pan2020multicell, soleymani2022improper}
 \begin{equation}\label{eq-ch}
\mathbf{h}_{lk,i}\left(\{\bm{\Theta}\}\right)=\underbrace{\sum_{m=1}^M\mathbf{f}_{lk,m}\bm{\Theta}_m\mathbf{G}_{mi}}_{\text{Links through RIS}}+
\underbrace{\mathbf{d}_{lk,i}}_{\text{Direct link}}\in \mathbb{C}^{1\times N_{BS}},
\end{equation}
where $\mathbf{d}_{lk,i}\in \mathbb{C}^{1\times N_{BS}}$ is the direct link between BS $i$ and u$_{lk}$, $\mathbf{G}_{mi}\in \mathbb{C}^{N_{RIS}\times N_{BS}}$ is the channel matrix between BS $i$ and RIS $m$, and $\mathbf{f}_{lk,m}\mathbb{C}^{1\times N_{RIS}}$ is the channel vector between RIS $m$ and u$_{lk}$.
The matrix  $\bm{\Theta}_m\in\mathbb{C}^{N_{RIS}\times N_{RIS}}$ contains the reflecting coefficients of RIS $m$ as
\begin{equation}
\bm{\Theta}_m=\text{diag}\left(\theta_{m1}, \theta_{m2},\cdots,\theta_{mN_{RIS}}\right),
\end{equation}
where each $\theta_{mi}$ for all $m$, $i$ is a complex-valued optimization parameter. 
We represent the set of reflecting coefficients $\{\bm{\Theta}_m\}_{m=1}^M$ as $\{\bm{\Theta}\}$. 
To simplify the notations, we drop the dependency of the channels with respect to $\{\bm{\Theta}\}$ unless it causes confusion. 
Note that  in this paper, we consider a regular RIS, which can only reflect signals. Indeed, both transceivers should be in the reflection space of a RIS in order to have a link through the RIS. If BS $i$ (or user u$_{lk}$) is not in the reflection space of RIS $m$, then we have $\mathbf{G}_{mi}=\mathbf{0}$ (or $\mathbf{f}_{lk,m}=\mathbf{0}$).
Moreover, note that we can optimize the channels only through  the reflecting coefficients $\theta_{mi}$ since the other parameters in \eqref{eq-ch} are given. The readers are referred to \cite{pan2020multicell, soleymani2022improper} for a  discussion on the large-scale and small-scale fading of the channels.
\subsubsection{Feasibility sets for the RIS components}
Here, we describe three common feasibility sets for the RIS components.
In the ideal case, the amplitude and the phase of each reflecting coefficient can be independently  optimized, which leads to the following feasibility set \cite[Eq. (11)]{wu2021intelligent}
\begin{equation}
\mathcal{T}_{U}=\left\{\theta_{mi}:|\theta_{mi}|^2\leq 1 \,\,\,\forall m,i\right\}.
\end{equation}
This feasibility set provides a  performance upper bound. 

In a more realistic feasibility set, often found in the literature, we 
can tune the phase  of each reflecting coefficient continuously between $0$ and $2\pi$, while the amplitude is fixed to $1$ \cite{di2020smart, wu2021intelligent}.  This feasibility set can be written as  
\begin{equation}
\mathcal{T}_{I}=\left\{\theta_{mi}:|\theta_{mi}|= 1 \,\,\,\forall m,i\right\}.
\end{equation}
This case is also referred to as intelligent reflecting surface (IRS)  to emphasize that there is only a passive phase-shifting beamforming at the intelligent surfaces \cite{di2020smart}. This feasibility set has been widely used in the studies related to RIS \cite{di2020smart, wu2021intelligent, wu2019intelligent, kammoun2020asymptotic, yu2020joint, pan2020multicell, zhang2020intelligent}.
Note that, in some studies such as \cite{kammoun2020asymptotic}, an attenuation factor for the reflecting coefficients is considered such that $|\theta_{mi}|=\eta< 1$. In this paper, we do not consider this attenuation  factor; however, our analysis can easily be adapted to the feasibility set in \cite{kammoun2020asymptotic}.

There is yet another feasibility set in which the amplitude is not constant, but instead it is a deterministic function of the phase as \cite{abeywickrama2020intelligent} 
\begin{equation}\label{eq=9}
0\leq |\theta|_{\min}\leq |\theta_{mi}|= f(\angle \theta_{mi})\leq |\theta|_{\max}\leq 1\,\,\,\forall m,i,
\end{equation}
where $f(\cdot)$ is
\begin{equation}
f(\angle \theta_{mi})= |\theta|_{\min}+( 1-|\theta|_{\min})\left(\frac{\sin\left(\angle \theta_{mi}-\phi\right)+1}{2}\right)^{\alpha},
\end{equation}
where $|\theta|_{\min}$, $\alpha$ and $\phi$ are non-negative constant values. According to this model, the amplitude and phase of RIS components are not independent optimization parameters. 
The feasibility set of $\theta_{mi}$ for this model is 
\begin{equation}
\mathcal{T}_{C}=\left\{\theta_{mi}:|\theta_{mi}|= f(\angle \theta_{mi}), \,\angle \theta_{mi}\in[-\pi,\pi]  \,\,\,\forall m,i\right\}.
\end{equation}
Note that in \cite{abeywickrama2020intelligent}, $|\theta|_{\max}=1$. Moreover, when $\alpha=0$, or equivalently $|\theta|_{\min}=1$, this feasibility set is equivalent to $\mathcal{T}_{I}$. Note than $\mathcal{T}_{I}\subset\mathcal{T}_U$ and $\mathcal{T}_{C}\subset\mathcal{T}_U$, which implies that the solution for the feasibility set $\mathcal{T}_U$ is an upper bound for the performance of the other feasibility sets.

\subsection{Signal model}
The broadcast signal from BS $l$ is 
\begin{equation}
\mathbf{x}_l=\sum_{k=1}^{2K}\mathbf{x}_{lk}\in\mathbb{C}^{N_{BS}\times 1},
\end{equation}
 which is a superposition of the signals intended for all  users associated to BS $l$. 
We assume that $\mathbf{x}_{lk}$s for all $l,k$ are uncorrelated zero-mean (and possibly improper) Gaussian signals. 
Let us remind the reader that the real and imaginary parts of improper signals are correlated and/or have unequal powers \cite{schreier2010statistical,adali2011complex}. To deal with impropriety, we employ the real-decomposition method, which is more suitable for  studying this scenario. By the real-decomposition method, all parameters are written in a real domain by considering the real and imaginary parts of a signal as separate variables. 
We represent the real decomposition of $\mathbf{x}_{lk}$ by $\underline{\mathbf{x}}_{lk}=\left[\begin{array}{cc}\mathfrak{R}\{\mathbf{x}_{lk}\}^T&\mathfrak{I}\{\mathbf{x}_{lk}\}^T\end{array}\right]^T$, where $\mathfrak{R}\{x\}$ and $\mathfrak{I}\{x\}$ take, respectively, the real and imaginary part of $x$. Moreover, we represent the transmit covariance matrix of $\underline{\mathbf{x}}_{lk}$ by $\mathbf{P}_{lk}=\mathbb{E}\left\{\underline{\mathbf{x}}_{lk}\underline{\mathbf{x}}_{lk}^T\right\}$.

The received signal for the user u$_{lk}$ is
\begin{subequations}\label{rec-sig}
\begin{align}
y_{lk}&=\sum_{i=1}^L\mathbf{h}_{lk,i}
\sum_{j=1}^{2K}\mathbf{x}_{ij}+\mathbf{n}_{lk}
\\&
=
\underbrace{\mathbf{h}_{lk,l}
\mathbf{x}_{lk}}_{\text{Desired signal}}+
\underbrace{\mathbf{h}_{lk,l}
\sum_{j=1,j\neq k}^{2K}\mathbf{x}_{lj}}_{\text{Intracell interference}}+
\underbrace{ \sum_{i=1,i\neq l}^L\mathbf{h}_{lk,i}
\mathbf{x}_{i}}_{\text{Intercell interference}}+
\underbrace{n_{lk}}_{\text{Noise}},
\end{align}
\end{subequations}
where 
$\mathbf{h}_{lk,l}\in\mathbb{C}^{1\times N_{BS}}$ is the  channel between BS $l$ and u$_{lk}$, and $n_{lk}\in\mathbb{C}$ is a zero-mean   additive  white Gaussian noise with variance $\sigma^2$. 
Using the real-decomposition method to model IGS, \eqref{rec-sig} becomes
\begin{align}\label{rec-sig=2}
\underline{\mathbf{y}}_{lk}
=
\underbrace{\underline{\mathbf{H}}_{lk,l}%\left(\{\bm{\Theta}\}\right)
\underline{\mathbf{x}}_{lk}}_{\text{Desired signal}}+
\underbrace{\underline{\mathbf{H}}_{lk,l}%\left(\{\bm{\Theta}\}\right)
\sum_{j=1,j\neq k}^{2K}\underline{\mathbf{x}}_{lj}}_{\text{Intracell interference}}+
\underbrace{ \sum_{i=1,i\neq l}^L\underline{\mathbf{H}}_{lk,i}%\left(\{\bm{\Theta}\}\right)
\underline{\mathbf{x}}_{i}}_{\text{Intercell interference}}+
\underbrace{\underline{\mathbf{n}}_{lk}}_{\text{Noise}},
\end{align}
%\end{subequations}
 where $\underline{\mathbf{y}}_{lk}=\left[\begin{array}{cc}\mathfrak{R}\{y_{lk}\}&\mathfrak{I}\{y_{lk}\}\end{array}\right]^T\in\mathbb{R}^{2\times 1}$, $\underline{\mathbf{n}}_{lk}=\left[\begin{array}{cc}\mathfrak{R}\{n_{lk}\}&\mathfrak{I}\{n_{lk}\}\end{array}\right]^T\in\mathbb{R}^{2\times 1}$, and the channels in the real-decomposition method, $\underline{\mathbf{H}}_{lk,i}%\left(\cdot\right)
\in\mathbb{R}^{2\times 2N_{BS}}$ for all $i,l,k$, are given by
 \begin{equation}
 \underline{\mathbf{H}}_{lk,i}%\left(\cdot\right)
=\left[\begin{array}{cc}\mathfrak{R}\{\mathbf{h}_{lk,i}\}&-\mathfrak{I}\{\mathbf{h}_{lk,i}\}\\
\mathfrak{I}\{\mathbf{h}_{lk,i}\}&\mathfrak{R}\{\mathbf{h}_{lk,i}\}
\end{array}\right].
 \end{equation}

\subsection{NOMA}

We assume that the intercell interference (ICI) is treated as noise. Note that treating ICI as noise is a practical assumption since it might be very inefficient to decode and cancel ICI in a multicell BC with many users per cell \cite{cui2018qoe, lei2019load, ni2021resource, guo2021qos, you2020note, wang2019user, you2018resource}. The reason is that, due to the decodability constraint, the data transmission rates may be restricted by the weakest link between a BS and all the users (either inside or outside of the BS cell), which may significantly degrade the performance. Moreover, it may add unnecessary complexities to  the system design such as finding the optimal decoding order of the ICI messages and highly increase the signaling overheads of the system. It should be noted that treating ICI as noise does not mean that ICI is overlooked in the system design. In this paper, we jointly employ IGS and RIS to manage ICI by a careful design of the transmit covariance matrices and RIS components.

We, additionally, employ a NOMA-based technique to handle a part of the intracell interference. 
Employing a full SIC scheme  is very challenging even in a single-cell RIS-assisted broadcast channel \cite{yang2020intelligent, zeng2020sum}. The optimal user ordering in a single-cell RIS-assisted BC is conducted based on the effective signal-to-interference-plus-noise ratio (SINR), which depends on channel gains. 
However, in RIS-assisted systems, the effective SINR and consequently the user ordering depend on the channel coefficients, which in turn depend on the  reflecting coefficients \cite{yang2020intelligent}. Additionally, user ordering in a multicell BC is a difficult task since the users in a cell experience different levels of intercell interference, thus making it even more complicated to find the optimal user ordering in a multicell RIS-assisted BC, and 
to the best of our knowledge, the  optimal decoding strategy for such systems is unknown. 
Moreover, assuming $2K$ users per cell, there are $L\times 2K!$ possible user orderings in the multicell BC, so an exhaustive search is impossible from a practical point of view \cite{yang2020intelligent, zeng2020sum}. In addition to the complexities involved with finding the optimal decoding order, there are also other challenges to apply NOMA to a large number of users. For instance, the receives should be able to decode and cancel the signals of multiple users, which may increase the complexities and cost of devices. Thus, it may not be realistic to employ a full NOMA scheme especially when the number of users per cell grows, which motivates employing hybrid schemes \cite{liu2022evolution, ding2015impact, ding2017survey, dai2018survey, ding2020unveiling}. Note that  hybrid schemes can be very practical and efficient and may provide a suitable trade off between the system performance and complexity \cite{ding2017survey}.
Therefore, in this paper, we employ a hybrid scheme in which the users in each cell are divided into clusters with two users  each, and NOMA is applied only to the pairs of users. The optimal user ordering/pairing depends on the channel coefficients as well as the ICI level at the users for a multicell BC, and to the best of our knowledge, the optimal strategy for a multi-cell MISO RIS-assisted BC is unknown. 

\subsubsection{Proposed user ordering/pairing scheme}
In this paper, we assume that the user pairing is given, and we do not tackle obtaining the optimal user pairing/ordering. Instead, we focus on different aspects of RISs as well as on the resulting optimization problems, and we leave the optimal user pairing/ordering scheme for a future study. Note that a suboptimal approach for user pairing can be to pair users based on their distance to the serving BS \cite{nasir2020signal,tuan2019non}. For instance, in \cite{tuan2019non}, it was assumed that users in each cell belong to two groups with the same number of members in each group. They can be called cell-centric  users (CCUs) and cell-edge  users (CEUs). Then each CCU is randomly paired with a CEU. Due to the large-scale fading, it is reasonable to assume that CEUs experience a higher ICI level  than CCUs in systems without RIS. Additionally,  the direct link between the BS and one of its cell-edge serving user is weaker than the direct link for a CCU associated to the BS. It should be noted that such assumptions may not hold in RIS-assisted systems since RISs can modulate the channels, which implies this pairing scheme is not necessarily optimal in a multi-cell MISO RIS-assisted BC. However, this pairing scheme can be practical and involves much lower complexities than the optimal strategy.

\subsubsection{Rate expressions for IGS}
In the following, we consider the optimization problems for a given user pairing policy.
We represent the user with a higher decoding order with $k$ and its paired user with $\bar{k}=K+k$. According to the pairing scheme, user $k$ associated to BS $l$, u$_{lk}$, first decodes the signal of user $\bar{k}$ associated to BS $l$, u$_{l\bar{k}}$, and then subtract it from its received signal. In other words, the data intended for u$_{l\bar{k}}$ should be transmitted in a data rate that is decodable by both u$_{lk}$ and u$_{l\bar{k}}$.
Thus, the rate of user $\bar{k}$ associated to BS $l$, u$_{l\bar{k}}$, is \cite{tuan2019non}
\begin{equation}\label{eq-12}
r_{l\bar{k}}=\min\left(\bar{r}_{l\bar{k}},\bar{r}_{lk\rightarrow\bar{k}}\right),
\end{equation}
where $\bar{r}_{l\bar{k}}$ is the rate at  u$_{l\bar{k}}$ when trying to decode its own signal:
\begin{align}\label{eq-22}
\bar{r}_{l\bar{k}}&=\frac{1}{2}\log_2\left|{\bf I}+
{\bf D}^{-1}_{l\bar{k}}%\left(\{\boldsymbol{\theta}\},\{{\bf P}\}\right)
 \underline{{\bf H}}_{l\bar{k},l}%(\{\bm{\Theta}\}) 
{\bf P}_{l\bar{k}}\underline{{\bf H}}_{l\bar{k},l}^T%(\{\bm{\Theta}\}) 
 \right|
\\ &
=
\underbrace{
\frac{1}{2}\log_2\left|{\bf D}_{l\bar{k}}%\left(\{\boldsymbol{\theta}\},\{{\bf P}\}\right) 
+
\underline{{\bf H}}_{l\bar{k},l}%(\{\bm{\Theta}\}) 
{\bf P}_{l\bar{k}}\underline{{\bf H}}_{l\bar{k},l}^T%(\{\bm{\Theta}\}) 
\right|
}_{\bar{r}_{l\bar{k},1}}
-
\underbrace{
\frac{1}{2}\log_2\left|{\bf D}_{l\bar{k}}%\left(\{\boldsymbol{\theta}\},\{{\bf P}\}\right) 
\right|
}_{\bar{r}_{l\bar{k},2}},
\label{eq-22-2}
\end{align}
where 
$\mathbf{D}_{l\bar{k}}$ is the interference-plus-noise covariance matrix at u$_{l\bar{k}}$
\begin{equation}%{multline}
\label{eq=15}
\mathbf{D}_{l\bar{k}}%\left(\cdot\right) 
=
\underbrace{
\sum_{i=1,i \neq l}^L\underline{\mathbf{H}}_{l\bar{k},i}%\left(\{\bm{\Theta}\}\right)
\mathbf{P}_{i}\underline{\mathbf{H}}_{l\bar{k},i}^T%\left(\{\bm{\Theta}\}\right)
}_{\text{Intercell interference}}
+
\underbrace{
\sum_{j= 1,j\neq \bar{k}}^{2K}
\underline{\mathbf{H}}_{l\bar{k},l}%\left(\{\bm{\Theta}\}\right)
\mathbf{P}_{lj}\underline{\mathbf{H}}_{l\bar{k},l}^T%\left(\{\bm{\Theta}\}\right)
}_{\text{Intracell interference}}
+
\underbrace{
\frac{\sigma^2}{2}\mathbf{I}
}_{\text{Noise}},
\end{equation}%{multline}
where $\mathbf{P}_{i}=\sum_{j=1}^{2K}\mathbf{P}_{ij}$.
Note that the intracell interference in \eqref{eq=15} consists of two terms: intracell-intracluster and intracell-intercluster interference. The intracell-intracluster interference consists of the signal intended for the paired user $k$, i.e., $\underline{\mathbf{H}}_{l\bar{k},i}\mathbf{P}_{i}\underline{\mathbf{H}}_{l\bar{k},i}^T$, while the intracell-intercluster interference term consists of the signal intended for users of the other clusters in cell $l$.
Moreover, the term $\bar{r}_{lk\rightarrow\bar{k}}$ is the rate at u$_{lk}$ when trying to decode the signal intended for u$_{l\bar{k}}$:
\begin{align}\label{eq-25}
\bar{r}_{lk\rightarrow\bar{k}}&=\frac{1}{2}\log_2\left|{\bf I}+
\left[{\bf D}_{lk}%\left(\{\boldsymbol{\theta}\},\{{\bf P}\}\right)
+
\underline{{\bf H}}_{lk,l}%(\{\bm{\Theta}\}) 
{\bf P}_{lk}\underline{{\bf H}}_{lk,l}^T%(\{\bm{\Theta}\}) 
\right]^{-1}
 \underline{{\bf H}}_{lk,l}%(\{\bm{\Theta}\}) 
{\bf P}_{l\bar{k}}\underline{{\bf H}}_{lk,l}^T%(\{\bm{\Theta}\})  
\right|\\
 &=
\underbrace{
\frac{1}{2}\log_2\left|{\bf D}_{lk}%\left(\{\boldsymbol{\theta}\},\{{\bf P}\}\right) 
+
\underline{{\bf H}}_{lk,l}%(\{\bm{\Theta}\}) 
{\bf P}_{lk}\underline{{\bf H}}_{lk,l}^T%(\{\bm{\Theta}\}) 
+
\underline{{\bf H}}_{lk,l}%(\{\bm{\Theta}\}) 
{\bf P}_{l\bar{k}}\underline{{\bf H}}_{lk,l}^T%(\{\bm{\Theta}\}) 
\right|
}_{\bar{r}_{lk\rightarrow\bar{k},1}%\left(\cdot\right)
}
\nonumber
\\
&\,\,\,\,\,
-
\underbrace{
\frac{1}{2}\log_2\left|{\bf D}_{lk}%\left(\{\boldsymbol{\theta}\},\{{\bf P}\}\right) 
+
\underline{{\bf H}}_{lk,l}%(\{\bm{\Theta}\}) 
{\bf P}_{lk}\underline{{\bf H}}_{lk,l}^T%(\{\bm{\Theta}\}) 
\right|
}_{\bar{r}_{lk\rightarrow\bar{k},2}%\left(\cdot\right)
},
%\nonumber,
\label{eq-26}
\end{align}
where ${\bf D}_{lk}$ can similarly be written as 
\begin{equation}%{multline}
\mathbf{D}_{lk}%\left(\cdot\right) 
=
\underbrace{
\sum_{i=1,i \neq l}^L\underline{\mathbf{H}}_{lk,i}%\left(\{\bm{\Theta}\}\right)
\mathbf{P}_{i}\underline{\mathbf{H}}_{lk,i}^T%\left(\{\bm{\Theta}\}\right)
}_{\text{Intercell interference}}
+
\underbrace{
\sum_{j= 1,j\neq k,\bar{k}}^{2K}
\underline{\mathbf{H}}_{lk,l}%\left(\{\bm{\Theta}\}\right)
\mathbf{P}_{lj}
\underline{\mathbf{H}}_{lk,l}^T%\left(\{\bm{\Theta}\}\right)
}_{\text{Intracell interference}}
+
\underbrace{
\frac{\sigma^2}{2}\mathbf{I}
}_{\text{Noise}},
\end{equation}%{multline}
where the intracell interference consists of only the intracell-intercluster interference since the intracell-intracluster interference is canceled by employing SIC. 
Note that by transmitting data to u$_{l\bar{k}}$ at rate $r_{l\bar{k}}$ in \eqref{eq-12}, it is ensured that the data is decodable at both u$_{l{k}}$ and u$_{l\bar{k}}$ users.  
Additionally, the rate for u$_{lk}$ is
\begin{align}\label{eq-28}
r_{lk}&=\frac{1}{2}\log_2\left|{\bf I}+
{\bf D}_{lk}^{-1}%\left(\cdot\right)
 \underline{{\bf H}}_{lk,l}%(\{\bm{\Theta}\}) 
{\bf P}_{lk}\underline{{\bf H}}_{lk,l}^T%(\{\bm{\Theta}\})  
\right|
\\ &
=
\underbrace{
\frac{1}{2}\log_2\left|{\bf D}_{lk}%\left(\cdot\right) 
+
\underline{{\bf H}}_{lk,l}%(\{\bm{\Theta}\}) 
{\bf P}_{lk}\underline{{\bf H}}_{lk,l}^T%(\{\bm{\Theta}\}) 
\right|
}_{\tilde{r}_{lk,1}%\left(\cdot\right) 
}
-
\underbrace{
\frac{1}{2}\log_2\left|{\bf D}_{lk}%\left(\cdot\right) 
\right|}
_{\tilde{r}_{lk,2}%\left(\cdot\right) 
}.\label{eq-29}
\end{align}
Note that we write the rates in two different formats so that  we can use the more suitable format for each optimization problem. For instance, to optimize over transmit covariance matrices it is more convenient to write the rates as a difference of two concave functions in $\{{\bf P}\}=\{{\bf P}_{lk},\forall l,k\}$ since it allows us to employ majorization-minimization (MM)  algorithms. 
However, for optimizing over the reflecting coefficients, we employ the rate expressions in \eqref{eq-22-2}, \eqref{eq-26} and \eqref{eq-29}, as will be discussed in the following sections.

\subsubsection{Energy-efficiency}
The EE of each user is defined as the ratio between its achievable rate and the power intended for the data transmission to the user, i.e., \cite{zappone2015energy, buzzi2016survey}
\begin{equation}
EE_{lk}=\frac{r_{lk}}
{P_c+\eta\text{Tr}\left(\mathbf{P}_{lk}\right)},
\end{equation}
where $\eta^{-1}$ is the power transmission
efficiency of each BS, and $P_c$ is the constant power consumption for transmitting data to a user, given by \cite[Eq. (27)]{soleymani2022improper}.

\section{Maximizing the Minimum-Weighted Rate}\label{sec-ii}
As a figure of merit for the spectral efficiency and coverage, we maximize the  minimum-weighted rate, which can be written as
\begin{align}\label{opt-prob-wmrm}
\underset{
r,\{\bm{\Theta}\}\in\mathcal{T}, \{\mathbf{P}\}\in\mathcal{P}}
{\max} \hspace{0.2cm}&r,
&
\text{s.t.}\hspace{0.4cm}& \lambda_{lk}r_{lk}\geq r, &\forall l,k,
\end{align}
where $\lambda_{lk}$ is the corresponding non-negative weight, $r_{l\bar{k}}=\min\left(\bar{r}_{l\bar{k}},\bar{r}_{lk\rightarrow\bar{k}}\right)$, and $\mathcal{P}$ denotes the feasibility set of the transmit covariance matrices, which
for the IGS and PGS schemes  are, respectively, \cite{soleymani2022improper}
\begin{align}
\mathcal{P}_{I}&\!=\left\{\{\mathbf{P}\}:\sum_{k=1}^{2K}\text{Tr}\left(\mathbf{P}_{lk}\right)\leq p_l, \mathbf{P}_{lk}\succcurlyeq\mathbf{0}, \forall l,k\right\},\\
\mathcal{P}_{P}&\!=\!\!\left\{\!\!\{\mathbf{P}\}\!:\!\!\!\sum_{k=1}^{2K}\!\text{Tr}\!\left(\mathbf{P}_{lk}\right)\!\!\leq p_l, \mathbf{P}_{lk}=\mathbf{P}_{t},\mathbf{P}_{lk}\succcurlyeq\mathbf{0}, \forall l,k\!\!\right\}\!\!,
\end{align}
where $p_l$ is the power budget of BS $l$, and $\mathbf{P}_{t}$ is a $2\times 2$ block matrix that fulfills the structure in \cite[Eq. (5)]{soleymani2022improper}. As can be verified through (23) and (24), PGS is a special case of IGS, and thus, the optimal IGS scheme never performs worse than the optimal PGS scheme.  Readers can refer to  \cite{schreier2010statistical,adali2011complex} for further details on improper signaling. 
Note that our proposed algorithms can be applied to both IGS and PGS schemes, and to write the equations for the most general case, we represent the feasibility set by $\mathcal{P}$.
Moreover, note that employing the rate profile technique, we can characterize the rate region by solving \eqref{opt-prob-wmrm} and varying the weights for all possible $\lambda_{lk}$, satisfying $\sum_{\forall l}\sum_{\forall k}\lambda_{lk}^{-1}=1$ \cite{soleymani2022improper}.
Additionally, the weights $\lambda_{lk}$ can be chosen based on the priorities of the users \cite{soleymani2022improper}. 

Unfortunately, \eqref{opt-prob-wmrm} is a non-convex difficult programming problem, and it is complicated to find its global optimal solution. 
Hence, we propose a suboptimal scheme by employing an iterative alternating optimization (AO) approach. 
That is, at iteration $t$, we first fix the reflecting coefficients to the solution of the previous step, $\{\bm{\Theta}^{(t-1)}\}$, and update the covariance matrices, $\{\mathbf{P}^{(t)}\}$. 
We then fix the covariance matrices $\{\mathbf{P}^{(t)}\}$ and optimize over reflecting coefficients to obtain $\{\bm{\Theta}^{(t)}\}$. 
We iterate this approach until a convergence criterion is met. 
Note that the convergence point of our algorithm may depend on the initial point since it is an iterative algorithm, based on MM. 
The initial point should be feasible and can be chosen either randomly or based on a heuristic approach \cite{lipp2016variations}. To obtain a better performance, we can run the algorithms for several initial points and choose the best solution among them  \cite{lipp2016variations}.  
It should be noted that the proposed algorithms actively optimize the transmission parameters at BSs as well as the environment  through RIS. Indeed, each iteration consists of two steps. The first step is to optimize the transmit covariance matrices, which are transmission parameters at BSs. The second step is to optimize the RIS components, which modify the environment.  

Unfortunately, the corresponding optimization problems are  still complicated even when we fix covariance matrices or reflecting coefficients.  
To tackle these problems, we employ MM and propose suboptimal algorithms to solve the optimization problems.

\subsection{Optimization of the transmit covariance matrices} 
In this subsection, we obtain transmit covariance matrices $\{\mathbf{P}^{(t)}\}$ for given reflecting coefficients $\{\bm{\Theta}^{(t-1)}\}$ by solving
% \begin{subequations}\label{opt-prob-wmrm-2}
\begin{align}\label{opt-prob-wmrm-2}
\underset{
r, \{\mathbf{P}\}\in\mathcal{P}}
{\max} \hspace{0.2cm}&r,
&
\text{s.t.}\hspace{0.4cm}& \lambda_{lk}r_{lk}\geq r, &\forall l,k.
%\\
%&&&\lambda_{l\bar{k}}r_{l\bar{k}}\geq r, &\forall l,\bar{k}.
\end{align}
%\end{subequations}
%where $\{\bm{\Theta}^{(t-1)}\}$ is the solution of the alternating algorithm at the previous step $t-1$, and $t$ is the iteration number. 
To solve this non-convex problem, we can apply difference of convex programming (DCP) to obtain a suboptimal solution for \eqref{opt-prob-wmrm-2}. 
The reason is that the rates can be written as a difference of two concave functions. 
Note that DCP falls into MM algorithms and obtains a stationary point of the original problem \cite{sun2017majorization}. 
DCP is based on the first-order Taylor expansion that approximates the convex part of the rates with an affine (linear) function. %, which is also known as convex-concave procedure (CCP). 
To find a suitable surrogate function for the rates, we employ the results in the following lemma.
%Due to space restriction, we briefly describe the optimization of covariance matrices in the following.
\begin{lemma}\label{lem-1} 
Consider the function $f(\mathbf{P})=\log\det(\mathbf{A}+\mathbf{B}\mathbf{P}\mathbf{B}^T),$ where $\mathbf{A}\in\mathbb{R}^{N\times N}$ and $\mathbf{B}\in\mathbb{R}^{N\times M}$ are constant matrices. Moreover, $\mathbf{A}$ and $\mathbf{P}$ are $M \times M$ real positive semi-definite matrices. 
Using the first-order Taylor expansion, an affine upper bound for $f(\mathbf{P})$ can be found as
\begin{equation*}%\label{eq-lem-1}
f(\mathbf{P}) \leq f(\mathbf{P}^{(t)})  + \text{\em{Tr}}\left(\mathbf{B}^T(\mathbf{A}+\mathbf{B}\mathbf{P}^{(t)}\mathbf{B}^T)^{-1}\mathbf{B}(\mathbf{P}-\mathbf{P}^{(t)})\right), 
\end{equation*}
where $\mathbf{P}^{(t)}$ is any feasible fixed point.
\end{lemma}
\begin{proof}
Employing the first-order Taylor expansion, we have
\begin{equation*}%\label{eq-lem-1}
f(\mathbf{P}) \leq f(\mathbf{P}^{(t)})  + \text{{Tr}}\left(\left[\frac{\partial f(\mathbf{P})}{\partial \mathbf{P}}|_{\mathbf{P}^{(t)}}\right]^T(\mathbf{P}-\mathbf{P}^{(t)})\right), 
\end{equation*}
where $\frac{\partial f(\mathbf{P})}{\partial \mathbf{P}}|_{\mathbf{P}^{(t)}}$ is the derivative of $f(\mathbf{P})$ with respect to $\mathbf{P}$ at $\mathbf{P}^{(t)}$, which is given by
\begin{equation*}
\frac{\partial f(\mathbf{P})}{\partial \mathbf{P}}|_{\mathbf{P}^{(t)}}=\mathbf{B}^T(\mathbf{A}+\mathbf{B}\mathbf{P}^{(t)}\mathbf{B}^T)^{-1}\mathbf{B}.
\end{equation*}
Note that $\frac{\partial f(\mathbf{P})}{\partial \mathbf{P}}$ is symmetric, thus we have 
$\left[\frac{\partial f(\mathbf{P})}{\partial \mathbf{P}}|_{\mathbf{P}^{(t)}}\right]^T=\frac{\partial f(\mathbf{P})}{\partial \mathbf{P}}|_{\mathbf{P}^{(t)}}$.
\end{proof}
A suitable concave lower bound for rates can be obtained by employing Lemma \ref{lem-1}. 
That is, we keep the concave part of the rate functions and approximate the convex part by a linear lower-bound function using Lemma \ref{lem-1}. 
Thus, the surrogate function for the rate of u$_{lk}$, i.e., the the user with a higher order,  can be written as 
\begin{multline}%{equation}
\label{l-r-lk}
 r_{lk}\geq \tilde{r}_{lk}=r_{lk,1}\left(\{\mathbf{P}\},
 \{\bm{\Theta}^{(t-1)}\}\right) 
-r_{lk,2}^{(t-1)}
-
\\
\sum_{i=1}^L\sum_{j=1}^{2K}%\frac{\partial\tilde{r}_{lk,2}}{\partial\mathbf{P}_{ij}}
\text{Tr}\left(
\frac{
\underline{\mathbf{H}}_{lk,i}^T
(\mathbf{D}_{lk}^{(t-1)})^{-1}
%\left(\{\mathbf{P}\}, \{\bm{\Theta}^{(t-1)}\}\right) 
\underline{\mathbf{H}}_{lk,i}
}
{2\ln 2}
\left(\mathbf{P}_{ij}-\mathbf{P}_{ij}^{(t-1)}\right)
\right),
\end{multline}%{equation}
where $r_{lk,2}^{(t-1)}=r_{lk,2}\left(\{\mathbf{P}^{(t-1)}\},
 \{\bm{\Theta}^{(t-1)}\}\right)$, and $\mathbf{D}_{lk}^{(t-1)}=\mathbf{D}_{lk}\left(\{\mathbf{P}^{(t-1)}\},
 \{\bm{\Theta}^{(t-1)}\}\right)$.
Similarly, a concave lower-bound for the rate of u$_{l\bar{k}}$, i.e., the user paired to u$_{lk}$, can be found as
\begin{equation}\label{l-r-lk-bar-low-32}
\tilde{r}_{l\bar{k}}=\min\left( \tilde{r}_{l\bar{k}}^l,\tilde{r}_{lk\rightarrow\bar{k}}^l\right),
\end{equation}
where $\tilde{r}_{lk\rightarrow\bar{k}}^l$ is given by \eqref{l-r-lk-bar-2} on the top of the next page and
\begin{figure*}[t]
\begin{align}
\tilde{r}_{lk\rightarrow\bar{k}}^l&=\bar{r}_{lk\rightarrow\bar{k},1}\left(\{\mathbf{P}\},
 \{\bm{\Theta}^{(t-1)}\}\right) 
-\bar{r}_{lk\rightarrow\bar{k},2}^{(t-1)}
%\nonumber
%\\&\hspace{1cm}
-\sum_{i=1}^L\sum_{j=1}^{2K}
\text{Tr}\left(
\frac{
\underline{\mathbf{H}}_{lk,i}^T
\left(
\mathbf{D}_{lk}^{(t-1)}+\underline{\mathbf{H}}_{lk,i}
\mathbf{P}_{lk}^{(t-1)}
\underline{\mathbf{H}}_{lk,i}^T
\right)^{-1}
\underline{\mathbf{H}}_{lk,i}
}
{2\ln 2}
\left(\mathbf{P}_{ij}-\mathbf{P}_{ij}^{(t-1)}\right)
\right),
\label{l-r-lk-bar-2}
\\
\setcounter{equation}{32}
\nonumber
\hat{r}_{l\bar{k}}^l&=\bar{r}_{l\bar{k}}^{(t-1)} 
-\frac{1}{2\ln 2}\text{{ Tr}}\left(
\bar{\mathbf{V}}_{l\bar{k}}\bar{\mathbf{V}}_{l\bar{k}}^T\bar{\mathbf{Y}}_{l\bar{k}}^{-1}
\right)
+
\frac{1}{\ln 2}
\text{{ Tr}}\left(
\bar{\mathbf{V}}_{l\bar{k}}^T\bar{\mathbf{Y}}_{l\bar{k}}^{-1}%\mathbf{V}_{lk}
\underline{{\bf H}}_{l\bar{k},l}(\{\bm{\Theta}\}) {\bf P}_{l\bar{k}}^{1/2^{(t)}}
\right)
\\
%\nonumber
&\hspace{.4cm}-
\frac{1}{2\ln 2}
\text{{ Tr}}\left(
(\bar{\mathbf{Y}}^{-1}_{l\bar{k}}-(\bar{\mathbf{V}}_{l\bar{k}}\bar{\mathbf{V}}^T_{l\bar{k}} + \bar{\mathbf{Y}}_{l\bar{k}})^{-1})^T
\left({\bf D}_{l\bar{k}} + \underline{{\bf H}}_{l\bar{k},l}(\{\bm{\Theta}\}) {\bf P}_{l\bar{k}}^{(t)}\underline{{\bf H}}_{l\bar{k},l}^T(\{\bm{\Theta}\}) 
\right)
\right)
\label{eq=33=}
\\
\nonumber
\hat{r}_{lk\rightarrow\bar{k}}^l&
=
\bar{r}_{lk\rightarrow\bar{k}}^{(t-1)} 
-\frac{1}{2\ln 2}\text{{ Tr}}\left(
\bar{\mathbf{V}}_{lk\rightarrow\bar{k}}\bar{\mathbf{V}}_{lk\rightarrow\bar{k}}^T\bar{\mathbf{Y}}_{lk\rightarrow\bar{k}}^{-1}
\right)
+
\frac{1}{\ln 2}
\text{{ Tr}}\left(
\bar{\mathbf{V}}_{lk\rightarrow\bar{k}}^T\bar{\mathbf{Y}}_{lk\rightarrow\bar{k}}^{-1}\underline{{\bf H}}_{lk,l}(\{\bm{\Theta}\}) {\bf P}_{l\bar{k}}^{1/2^{(t)}}
\right)
\\
&\hspace{.4cm}
-
\frac{1}{2\ln 2}
\text{{ Tr}}\left(
(\bar{\mathbf{Y}}^{-1}_{lk\rightarrow\bar{k}}-(\bar{\mathbf{V}}_{lk\rightarrow\bar{k}}\bar{\mathbf{V}}^T_{lk\rightarrow\bar{k}} + \bar{\mathbf{Y}}_{lk\rightarrow\bar{k}})^{-1})^T
(\underline{{\bf H}}_{lk,l}(\{\bm{\Theta}\}) {\bf P}_{l\bar{k}}^{(t)}\underline{{\bf H}}_{lk,l}^T(\{\bm{\Theta}\}) 
+
\mathbf{Y}_{lk\rightarrow\bar{k}})
\right),
\label{eq=34=}
%\nonumber
\end{align}
\hrulefill 
\end{figure*}
\begin{align}
\setcounter{equation}{28}
\nonumber
%\bar{r}_{l\bar{k}}&\geq 
\tilde{r}_{l\bar{k}}^l&=\bar{r}_{l\bar{k},1}\left(\{\mathbf{P}\},
 \{\bm{\Theta}^{(t-1)}\}\right) 
-\bar{r}_{l\bar{k},2}^{(t-1)}
\\&
\label{l-r-lk-bar}
-\sum_{i=1}^L\sum_{j=1}^{2K}
\text{Tr}\left(
\frac{
\underline{\mathbf{H}}_{l\bar{k},i}^T
(\mathbf{D}_{l\bar{k}}^{(t-1)})^{-1}
\underline{\mathbf{H}}_{l\bar{k},i}
}
{2\ln 2}
\left(\mathbf{P}_{ij}-\mathbf{P}_{ij}^{(t-1)}\right)
\right),
\end{align}
where $r_{l\bar{k},2}^{(t-1)}=r_{l\bar{k},2}\left(\{\mathbf{P}^{(t-1)}\},
 \{\bm{\Theta}^{(t-1)}\}\right)$, \\
$r_{lk\rightarrow\bar{k},2}^{(t-1)}=r_{lk\rightarrow\bar{k},2}\left(\{\mathbf{P}^{(t-1)}\},
 \{\bm{\Theta}^{(t-1)}\}\right),$ and \\
$\mathbf{D}_{l\bar{k}}^{(t-1)}=\mathbf{D}_{l\bar{k}}\left(\{\mathbf{P}^{(t-1)}\},
 \{\bm{\Theta}^{(t-1)}\}\right)$. 
Substituting the concave lower bounds in \eqref{opt-prob-wmrm-2}, we have
\begin{align}\label{opt-prob-wmrm-3}
\underset{
r, \{\mathbf{P}\}\in\mathcal{P}}
{\max} \hspace{0.2cm}&r,
&
\text{s.t.}\hspace{0.4cm}& \lambda_{lk}\tilde{r}_{lk}\geq r, &\forall l,k.
%\\
%&&&\lambda_{l\bar{k}}r_{l\bar{k}}^l\geq r, &\forall l,\bar{k}.
\end{align}
This optimization problem is convex and can be efficiently solved to obtain a new set of transmit covariance matrices $\{\mathbf{P}^{(t)}\}$.

\subsection{Optimization of the reflecting coefficients} \label{sec-orm}

In this subsection, we update the reflecting coefficients while the covariance matrices are fixed to $\{\mathbf{P}^{(t)}\}$. 
That is, we solve the following optimization problem 
%\begin{subequations}\label{opt-prob-wmrm-rc}
\begin{align}\label{opt-prob-wmrm-rc}
\underset{
r,\{\bm{\Theta}\}\in\mathcal{T}}
{\max} \hspace{0.2cm}&r,
&
\text{s.t.}\hspace{0.4cm}& \lambda_{lk}r_{lk}\geq r, &\forall l,k.
\end{align}
%\end{subequations}
Unfortunately, this problem is non-convex and difficult too. To solve it, we first find a suitable concave lower bound for the rates by employing the following lemma. Then, we consider each feasibility set separately.
\begin{lemma}[\!\!\cite{yu2020improper, yu2020joint}]\label{lem-2} 
The following inequality holds for all  $N \times N$ positive definite matrices $\mathbf{Y}$ and $\bar{\mathbf{Y}}$, and any arbitrary  $N \times M$ matrices $\mathbf{V}$ and $\bar{\mathbf{V}}$:
\begin{multline}%{align} 
%\nonumber
\ln \left|\mathbf{I}+\mathbf{V}\mathbf{V}^H\mathbf{Y}^{-1}\right|
%&
\geq
 \ln \left|\mathbf{I}+\bar{\mathbf{V}}\bar{\mathbf{V}}^H\bar{\mathbf{Y}}^{-1}\right|
 \\-
\text{{\em Tr}}\left(
\bar{\mathbf{V}}\bar{\mathbf{V}}^H\bar{\mathbf{Y}}^{-1}
\right)
+
2\mathfrak{R}\left\{\text{{\em Tr}}\left(
\bar{\mathbf{V}}^H\bar{\mathbf{Y}}^{-1}\mathbf{V}
\right)\right\}\\
%&\hspace{.4cm}
-
\text{{\em Tr}}\left(
(\bar{\mathbf{Y}}^{-1}-(\bar{\mathbf{V}}\bar{\mathbf{V}}^H + \bar{\mathbf{Y}})^{-1})^H(\mathbf{V}\mathbf{V}^H+\mathbf{Y})
\right).
\label{lower-bound}
\end{multline}%{align}
\end{lemma}

\begin{corollary} \label{theo-1}
A concave lower-bound 
$\hat{r}_{lk}$ and $\hat{r}_{l\bar{k}}$ for, respectively, $r_{lk}$ and $r_{l\bar{k}}$ can be found as
\begin{align*}
\nonumber
\hat{r}_{lk}&=r_{lk}^{(t-1)}
-\frac{1}{2\ln 2}\text{{\em Tr}}\left(
\bar{\mathbf{V}}_{lk}\bar{\mathbf{V}}_{lk}^T\bar{\mathbf{Y}}_{lk}^{-1}
\right)
\\ \nonumber
&\,\,\,+
\frac{1}{\ln 2}
\text{{\em Tr}}\left(
\bar{\mathbf{V}_{lk}}^T\bar{\mathbf{Y}_{lk}}^{-1}%\mathbf{V}_{lk}
\underline{{\bf H}}_{lk,l}(\{\bm{\Theta}\}) {\bf P}_{lk}^{1/2^{(t)}}
\right)
\\
&\hspace{.1cm}-
\frac{1}{2\ln 2}
\text{{\em Tr}}\left(
(\bar{\mathbf{Y}}^{-1}_{lk}-(\bar{\mathbf{V}}_{lk}\bar{\mathbf{V}}^T_{lk} + \bar{\mathbf{Y}}_{lk})^{-1})^T
%(\mathbf{V}_{lk}\mathbf{V}_{lk}^H+\mathbf{Y}_{lk})
\left(
%{\bf D}_{lk}\left(\cdot\right) + \underline{{\bf H}}_{lk,l}(\{\boldsymbol{\theta}\}) {\bf P}_{lk}^{(t)}\underline{{\bf H}}_{lk,l}^T(\{\boldsymbol{\theta}\}) 
\mathbf{Y}_{lk\rightarrow\bar{k}}
\right)
\right)
\\
\hat{r}_{l\bar{k}}&=\min\left( \hat{r}_{l\bar{k}}^l,\hat{r}_{lk\rightarrow\bar{k}}^l\right),
\end{align*}
where $r_{lk}^{(t-1)}=r_{lk}\left(\{\mathbf{P}^{(t)}\},\{\bm{\Theta}^{(t-1)}\}\right)$. %is the rate of u$_{lk}$ at the beginning of this previous step, 
Moreover, $\hat{r}_{l\bar{k}}^l$ and $\hat{r}_{lk\rightarrow\bar{k}}^l$ are, respectively, given by \eqref{eq=33=} and \eqref{eq=34=}, on the top of this page, where 
$\bar{r}_{l\bar{k}}^{(t-1)}=\bar{r}_{l\bar{k}}\left(\{\mathbf{P}^{(t)}\},
 \{\bm{\Theta}^{(t-1)}\}\right)$ and 
 $\bar{r}_{lk\rightarrow\bar{k}}^{(t-1)} = \bar{r}_{lk\rightarrow\bar{k}}\left(\{\mathbf{P}^{(t)}\},
 \{\bm{\Theta}^{(t-1)}\}\right) $.
Furthermore, the other parameters are defined as %in the following
$\bar{\mathbf{V}}_{lk}=\underline{\mathbf{H}}_{lk,l}\left(\{\bm{\Theta}^{(t-1)}\}\right)\mathbf{P}_{lk}^{(t)^{1/2}}$, 
%$\bar{\mathbf{V}}_{l\bar{k}}=\underline{\mathbf{H}}_{l\bar{k},l}\left(\{\bm{\Theta}^{(t-1)}\}\right)\mathbf{P}_{l\bar{k}}^{(t)^{1/2}}$, 
%$\bar{\mathbf{V}}_{lk\rightarrow\bar{k}}=\underline{\mathbf{H}}_{lk,l}\left(\{\bm{\Theta}^{(t-1)}\}\right)\mathbf{P}_{l\bar{k}}^{(t)^{1/2}}$, $\mathbf{V}_{lk\rightarrow\bar{k}}=\underline{\mathbf{H}}_{lk,l}\left(\{\bm{\Theta}^{(t-1)}\}\right)\mathbf{P}_{l\bar{k}}^{1/2}$, 
$\bar{\mathbf{Y}}_{lk}=\mathbf{D}_{lk}\left(\{\mathbf{P}^{(t)}\},\{\bm{\Theta}^{(t-1)}\}\right)$, %$\bar{\mathbf{Y}}_{l\bar{k}}=\mathbf{D}_{l\bar{k}}\left(\{\mathbf{P}^{(t)}\},\{\bm{\Theta}^{(t-1)}\}\right)$, 
and 
%\begin{subequations}
\begin{align*}
%\bar{\mathbf{V}}_{lk}&=\underline{\mathbf{H}}_{lk,l}\left(\{\bm{\Theta}^{(t-1)}\}\right)\mathbf{P}_{lk}^{(t)^{1/2}},\\
%\bar{\mathbf{V}}_{l\bar{k}}&=\underline{\mathbf{H}}_{l\bar{k},l}\left(\{\bm{\Theta}^{(t-1)}\}\right)\mathbf{P}_{l\bar{k}}^{(t)^{1/2}},\\
%\bar{\mathbf{V}}_{lk\rightarrow\bar{k}}&=\underline{\mathbf{H}}_{lk,l}\left(\{\bm{\Theta}^{(t-1)}\}\right)\mathbf{P}_{l\bar{k}}^{(t)^{1/2}},
%\\
\mathbf{V}_{lk\rightarrow\bar{k}}&=\underline{\mathbf{H}}_{lk,l}\left(\{\bm{\Theta}^{(t-1)}\}\right)\mathbf{P}_{l\bar{k}}^{1/2},\\
 %\bar{\mathbf{Y}}_{lk}&=\mathbf{D}_{lk}\left(\{\mathbf{P}^{(t)}\},\{\bm{\Theta}^{(t-1)}\}\right), 
%\\ 
%\bar{\mathbf{Y}}_{l\bar{k}}&=\mathbf{D}_{l\bar{k}}\left(\{\mathbf{P}^{(t)}\},\{\bm{\Theta}^{(t-1)}\}\right),
%\\
\mathbf{Y}_{lk\rightarrow\bar{k}}&=\mathbf{D}_{lk}\left(\{\mathbf{P}^{(t)}\},\{\bm{\Theta}\}\right)
+
 \underline{{\bf H}}_{lk,l}(\{\bm{\Theta}\}) {\bf P}_{lk}^{(t)}\underline{{\bf H}}_{lk,l}^T(\{\bm{\Theta}\}) , 
\\
\bar{\mathbf{Y}}_{lk\rightarrow\bar{k}}&=\mathbf{D}_{lk}\left(\{\mathbf{P}^{(t)}\},\{\bm{\Theta}^{(t-1)}\}\right)
\\&+
 \underline{{\bf H}}_{lk,l}(\{\bm{\Theta}^{(t-1)}\}) {\bf P}_{lk}^{(t)}\underline{{\bf H}}_{lk,l}^T(\{\bm{\Theta}^{(t-1)}\}).
\end{align*}
%\end{subequations}
\end{corollary}
\begin{proof}
Substituting the corresponding parameters in Lemma \ref{lem-2} results in the lower bounds. 
Note that since we employ the real-decomposition method, we can substitute the Hermitian operator with the transpose operator. Additionally, we can remove operator $\mathfrak{R}$ because all the variables are real in the real-decomposition method.
\end{proof}

\subsubsection{Feasibility set $\mathcal{T}_{U}$} Since $\mathcal{T}_{U}$ is a convex set, the following %surrogate optimization 
problem is convex: 
%\begin{subequations}\label{opt-prob-wmrm-rc-u}
\setcounter{equation}{34}
\begin{align}\label{opt-prob-wmrm-rc-u}
\underset{
r,\{\bm{\Theta}\}\in\mathcal{T}_U}
{\max} \hspace{0.2cm}&r,
&
\text{s.t.}\hspace{0.4cm}& \lambda_{lk}\hat{r}_{lk}\geq r, &\forall l,k.
%\\
%&&&\lambda_{l\bar{k}}\hat{r}_{l\bar{k}}\geq r, &\forall l,\bar{k},
\end{align}
%\end{subequations}
Hence, %The optimization problem \eqref{opt-prob-wmrm-rc-u} can be efficiently solved by existing optimization tools such as CVX. 
we can solve \eqref{opt-prob-wmrm-rc-u} efficiently by existing convex optimization tools. 
Note that the algorithm for this feasibility set converges to a stationary point of \eqref{opt-prob-wmrm} since the surrogate functions fulfill the three conditions in \cite[Sec. III]{soleymani2020improper}. Our solution for $\mathcal{T}_U$ is summarized in Algorithm I. 
\begin{table}[htb]
\label{alg-wsrm}
\begin{tabular}{l}
%\rule{0.55\textwidth}{2pt}\par\\
%\rule{\hsize}{2pt}\par
\hline 
 \small{\textbf{Algorithm I} NOMA-based IGS algorithm in Sec. \ref{sec-ii} with $\mathcal{T}_U$.}\\
%\rule{0.55\textwidth}{1.5pt}\par\\
\hline 
\hspace{0.2cm}\small{\textbf{Initialization}}\\
\hspace{0.2cm}\small{Set $\epsilon$, %$L$,  
$t=1$, %convergence=0, 
 $\{\mathbf{P}\}=\{\mathbf{P}^{(0)}\}$, and$\{\bm{\Theta}\}=\{\bm{\Theta}^{(0)}\}$ }\\
%\hspace{0.2cm}\small{{\bf For} $k=1,\cdots,K$}\\
%\hspace{1.2cm}\small{$p_k=P_k$}\\
%\hspace{0.2cm}\small{\textbf{End (For)}}\\
\hline 
\hspace{0.2cm}%\small{\textbf{Repeat}}\\
\small{\textbf{While} $\left(\underset{\forall l,k}{\min}\,\lambda_{lk}r^{(t)}_{lk}-\underset{\forall l,k}{\min}\,\lambda_{lk}r^{(t-1)}_{lk}\right)/\underset{\forall l,k}{\min}\,\lambda_{lk}r^{(t-1)}_{lk}\geq\epsilon$ }\\ %and $t\leq T$ \textbf{do}
\hspace{.6cm}\small{{\bf Optimizing over} $\{\mathbf{P}\}$ {\bf by fixing} $\{\bm{\Theta}^{(t-1)}\}$}\\
\hspace{1.2cm}\small{Obtain $\tilde{r}_{lk}^{(t-1)}$ %\left(\{\mathbf{P}\},\{\bm{\Theta}^{(t-1)}\}\right)$ 
based on \eqref{l-r-lk}-\eqref{l-r-lk-bar}}\\ %and \eqref{l-r-lk-bar-low-32}}\\
\hspace{1.2cm}\small{Compute $\{\mathbf{P}^{(t)}\}$ by solving \eqref{opt-prob-wmrm-3}}\\
\hspace{.6cm}\small{{\bf Optimizing over} $\{\bm{\Theta}\}$ {\bf by fixing} $\{\mathbf{P}^{(t-1)}\}$}\\
\hspace{1.2cm}\small{Obtain $\hat{r}_{lk}^{(t-1)}$ %\left(\{\mathbf{P}^{(t)}\},\{\bm{\Theta}\}\right)$ 
based on Corollary 1}\\
\hspace{1.2cm}\small{Compute $\{\bm{\Theta}^{(t)}\}$ by solving \eqref{opt-prob-wmrm-rc-u}}\\
%\hspace{1.2cm}\small{\textbf{End (For)}}\\
%\hspace{1.2cm}\small{{\bf If} $\left(\sum_{k=1}^{K}R^{(l+1)}_k-\sum_{k=1}^{K}R^{(l)}_k\right)/\sum_{k=1}^{K}R^{(l)}_k<\epsilon$}\\
%\hspace{1.8cm}\small{convergence=1}\\
%\hspace{1.2cm}\small{\textbf{End (If)}}\\
\hspace{.6cm}\small{$t=t+1$}\\
\hspace{0.2cm}\small{\textbf{End (While)}}\\
\hspace{0.2cm}\small{{\bf Return} $\{\mathbf{P}^{\star}\}$ and $\{\bm{\Theta}^{\star}\}$.}\\
\hline %\rule{0.55\textwidth}{2pt}\par
\end{tabular}
\end{table}
\subsubsection{Feasibility set $\mathcal{T}_{I}$}
Unfortunately, $\mathcal{T}_{I}$ is not a convex set, which makes the resulting surrogate problem non-convex.
To address this issue, we rewrite the constraint $|\theta_{m_n}|=1$ as the two following constraints
\begin{align}\label{eq-43}
|\theta_{mn}|^2&\leq 1\\
|\theta_{mn}|^2&\geq 1,\label{eq-44}
\end{align}
for all $m,n$.
The constraint \eqref{eq-43} is a convex constraint, but we have to approximate the constraint \eqref{eq-44}. 
To this end, we approximate \eqref{eq-44} by the following affine lower bound 
\begin{equation}\label{eq=low}
|\theta_{mn}|^2\geq |\theta_{mn}^{(t-1)}|^2+2\mathfrak{R}\left(\theta_{mn}^{(t-1)}(\theta_{mn}-\theta_{mn}^{(t-1)})^*\right).
\end{equation}
Thus, a surrogate function for the constraint \eqref{eq-44} can be written as 
\begin{equation}\label{eq-50}
|\theta_{mn}^{(t-1)}|^2+2\mathfrak{R}\left(\theta_{mn}^{(t-1)}(\theta_{mn}-\theta_{mn}^{(t-1)})^*\right)\geq 1.
\end{equation}
To speed up convergence, we introduce a positive parameter, $\epsilon$, and rewrite \eqref{eq-50} as
\begin{equation}\label{eq-51}
|\theta_{mn}^{(t-1)}|^2+2\mathfrak{R}\left(\theta_{mn}^{(t-1)}(\theta_{mn}-\theta_{mn}^{(t-1)})^*\right)\geq 1-\epsilon.
\end{equation}
The inequality in \eqref{eq-51} is a linear constraint, which makes the following problem  convex:
\begin{subequations}\label{opt-prob-wmrm-rc-i}
\begin{align}
\underset{
r,\{\bm{\Theta}\}}
{\max} \hspace{0.2cm}&r,
&
\text{s.t.}\hspace{0.4cm}& \lambda_{lk}\hat{r}_{lk}\geq r, &\forall l,k,
%\\
%&&&\lambda_{l\bar{k}}\hat{r}_{l\bar{k}}\geq r, &\forall l,\bar{k},
\\
&&& \eqref{eq-51}, \eqref{eq-43}.
\end{align}
\end{subequations}
The optimization problem \eqref{opt-prob-wmrm-rc-i} can be efficiently solved. 
Let us represent the normalized solution of  \eqref{opt-prob-wmrm-rc-i} by  $\{\hat{\bm{\Theta}}\}$. 
We update the reflecting coefficients for  the feasibility set $\mathcal{T}_{I}$ as
\begin{equation}\label{eq-42}
\{\bm{\Theta}^{(t)}\}=
\left\{
\begin{array}{lcl}
\{\hat{\bm{\Theta}}\}&\text{if}&
\underset{\forall l,k}{\min}
\left\{
r_{lk}\left(
\{\mathbf{P}^{(t)}\},\{\hat{\bm{\Theta}}\}
\right)
\right\}
\geq
\\&&
\underset{\forall l,k}{\min}
\left\{
r_{lk}\left(
\{\mathbf{P}^{(t)}\},\{\bm{\Theta}^{(t-1)}\}
\right)
\right\}
\\
\{\bm{\Theta}^{(t-1)}\}&&\text{Otherwise}.
\end{array}
\right.
\end{equation}
This updating rule ensures that the algorithm generates a sequence of non-decreasing minimum weighted rates, which guarantees its convergence. 
Note that although our proposed algorithm converges, we do not make any claim on its optimality.
\subsubsection{Feasibility set $\mathcal{T}_{C}$} The feasibility set  $\mathcal{T}_{C}$ is not  convex. To convexify this feasibility set, we relax the relationship between the phase and amplitude of reflecting components and approximate \eqref{eq=9} as 
\begin{equation}\label{eq=9-2}
|\theta|_{\min}\leq |\theta_{mi}|\leq 1\,\,\,\forall m,i.
\end{equation}
The constraint $|\theta_{mi}|\leq 1$ is convex since it can be written as $|\theta_{mi}|^2\leq 1$. However, we also have the constraint $|\theta_{mi}|\geq |\theta|_{\min}$, or equivalently $|\theta_{mi}|^2\geq |\theta|_{\min}^2$, which is similar to the constraint \eqref{eq-44}.
Thus, we can apply a similar procedure and employ the lower bound in \eqref{eq=low}, which yields
\begin{equation}\label{eq-50-2}
|\theta_{mn}^{(t-1)}|^2+2\mathfrak{R}\left(\theta_{mn}^{(t-1)}(\theta_{mn}-\theta_{mn}^{(t-1)})^*\right)\geq |\theta|_{\min}^2.
\end{equation}
Hence, the surrogate optimization problem is 
\begin{subequations}\label{opt-prob-wmrm-rc-i-2}
\begin{align}
\underset{
r,\{\bm{\Theta}\}}
{\max} \hspace{0.2cm}&r,
&
\text{s.t.}\hspace{0.4cm}& \lambda_{lk}\hat{r}_{lk}\geq r, &\forall l,k,
%\\
%&&&\lambda_{l\bar{k}}\hat{r}_{l\bar{k}}\geq r, &\forall l,\bar{k},
\\
&&& \eqref{eq-43}, \eqref{eq-50-2}.
\end{align}
\end{subequations}
Let us call the solution of \eqref{opt-prob-wmrm-rc-i-2} $\{\bm{\Theta}^{\star}\}$. We take the phases of  $\{\bm{\Theta}^{\star}\}$ to generate feasible reflecting coefficients by choosing 
\begin{equation}
\{\hat{\bm{\Theta}}\}=f(\angle\{\bm{\Theta}^{\star}\}).
\end{equation}
To make the solution of the algorithm non-decreasing, we update the reflecting coefficients  as \eqref{eq-42}.
%\begin{equation}
%\{\bm{\Theta}^{(t)}\}=
%\left\{
%\begin{array}{rcl}
%\{\hat{\bm{\Theta}}\}&\text{if}&
%\underset{\forall l,k}{\min}
%\left\{
%r_{lk}\left(
%\{\mathbf{P}^{(t)}\},\{\hat{\bm{\Theta}}\}
%\right)
%\right\}
%\geq\\&&
%\underset{\forall l,k}{\min}
%\left\{
%r_{lk}\left(
%\{\mathbf{P}^{(t)}\},\{\bm{\Theta}^{(t-1)}\}
%\right)
%\right\}
%\\
%\{\bm{\Theta}^{(t-1)}\}&&\text{Otherwise}.
%\end{array}
%\right.
%\end{equation}
%This update policy ensures that the 
Since the algorithm generates a sequence of non-decreasing minimum-weighted rates,  the convergence is guaranteed. However, the proposed algorithm for $\mathcal{T}_{C}$ does not necessarily converge to a stationary point of \eqref{opt-prob-wmrm}.

\subsection{Discussion on computational complexity} \label{sec=com=}
Since our proposed algorithm to solve the minimum-weighted-rate maximization (MWRM) problem is iterative, its computational complexity depends on the total number of iterations and the complexities involved with obtaining a solution in each iteration.
Note that such computational complexities highly depend on the implementation of our algorithms. In this subsection, we provide a brief discussion on the computational complexity of the proposed algorithm and provide an approximation for the upper bound of the number of multiplications to obtain a solution of our algorithm. 

Our proposed algorithm is iterative, and each iteration consists of two steps. In the first step, we obtain the transmit covariance matrices by solving \eqref{opt-prob-wmrm-3}, which is a convex optimization problem. According to \cite[Chapter 11]{boyd2004convex}, the total number of Newton steps to solve a convex optimization problem grows with the square root of the number of inequalities in the problem. There are $2KL$ inequalities in \eqref{opt-prob-wmrm-3}. To compute each lower bound in \eqref{l-r-lk}, there are approximately $8+24LN_{BS}^2$ multiplications. Moreover, there are $2LK$ lower bounds for the rates since the total number of users is $2KL$. 
Thus, the complexities involved with solving  \eqref{opt-prob-wmrm-3} can be approximated as $\mathcal{O}\left(\sqrt{KL}\left(L^2KN_{BS}^2\right)\right)$. 
In the second step, we update the RIS components. We consider three different feasibility sets for the RIS components. Here, we present the computational complexities to run our proposed algorithm only for the feasibility $\mathcal{T}_{I}$ due to a space restriction. It is very straightforward to repeat the analysis for the other feasibility sets.
To update the RIS components for  $\mathcal{T}_{I}$, we solve the convex optimization problem \eqref{opt-prob-wmrm-rc-i}. The number of constraints in \eqref{opt-prob-wmrm-rc-i} is $2LK+2MN_{RIS}$. Additionally, the approximate number of multiplications to obtain each lower bound $\hat{r}_{lk}$ can be approximated as $\mathcal{O}\left(N^2_{BS}L+MN_{RIS}N_{BS}\right)$. Thus, the computational complexity to solve \eqref{opt-prob-wmrm-rc-i} can be approximated as $\mathcal{O}\left(LK\sqrt{KL+MN_{RIS}}\left(N^2_{BS}L+MN_{RIS}N_{BS}\right)\right)$.
We set the maximum number of iterations to $N$. 
Hence, the total computational complexity of our proposed algorithm can be approximated as
\begin{multline*}
 \mathcal{O}\left(N\left(\sqrt{KL}\left(L^2KN_{BS}^2\right)\right.\right.\\
 +\left.\left. LK\sqrt{KL+MN_{RIS}}\left(N^2_{BS}L+MN_{RIS}N_{BS}\right)\right)\right).
 \end{multline*}

\section{Maximizing the Minimum Weighted EE}\label{EE-sec}
In this section, we solve the minimum-weighted-EE maximization (MWEEM) problem, i.e.,  
\begin{subequations}\label{ee-opt}
\begin{align}
\max_{\{\bm{\Theta}\}\in\mathcal{T},\{\mathbf{P}\}\in\mathcal{P},e}\,\,& e &
\text{s.t.}\hspace{.5cm}&\lambda_{lk} EE_{lk}\geq  e,&\forall l,k
\\
&&& r_{lk}\geq r^{th}_{lk}, &\forall l,k,
\label{qos-55-b}
%\\&&& r_{l\bar{k}}\geq r^{th}_{l\bar{k}}, &\forall l,\bar{k},\label{qos-55-c}
\end{align}
\end{subequations}
where the constraint \eqref{qos-55-b} %and \eqref{qos-55-c} are 
is related to the rate constraint for all users, and $ r^{th}_{lk}$ is the corresponding minimum required rate for the user. Similar to the MWRM problem, we can characterize the EE region by employing the EE profile technique. That is, we solve \eqref{ee-opt} and vary the weights for all possible $\lambda_{lk}$, satisfying $\sum_{\forall l}\sum_{\forall k}\lambda_{lk}^{-1}=1$ \cite{Sole1909:Energy}.
Moreover, similar to the MWRM problem, the non-negative weights $\lambda_{lk}$ can be chosen based on the priorities of the users \cite{soleymani2022improper}. 
This problem is not convex, and is even more complicated than the MWRM problem since EE functions have fractional structures in $\{\mathbf{P}\}$, and \eqref{ee-opt} is actually a multiple-fractional problem. 
To solve \eqref{ee-opt}, we employ an AO approach similar to the proposed scheme in Section \ref{sec-ii}. 
That is, we first fix the reflecting coefficient matrices and solve \eqref{ee-opt} by optimizing over the covariance matrices. 
To this end, we employ MM alongside with the generalized Dinkelbach algorithm (GDA). 
Note that GDA is an iterative algorithm to solve fractional programming problems. 
After updating the covariance matrices, we fix them and solve \eqref{ee-opt} over the reflecting-coefficient matrices. 
The resulting problem has a structure similar to \eqref{opt-prob-wmrm-rc}, and hence we can employ a similar approach to solve it. However, due to the rate constraint and the format of EE functions, we have to modify the constraints.

\subsection{Optimization of the transmit covariance matrices}
In this subsection,  we solve \eqref{ee-opt} for a fixed $\{\bm{\Theta}^{(t-1)}\}$, i.e.,   
\begin{subequations}\label{ee-opt-2}
\begin{align}
\max_{\{\mathbf{P}\}\in\mathcal{P},e}\,\,& e &
\text{s.t.}\hspace{.5cm}&\lambda_{lk}EE_{lk}\geq e,&\forall l,k
\\
&&& r_{lk}\geq r^{th}_{lk}, &\forall l,k.
%\label{qos-55-b}
%\\
%&&& r_{l\bar{k}}\geq r^{th}_{l\bar{k}}, &\forall l,\bar{k}.
%\label{qos-55-c}
\end{align}
\end{subequations}
This problem is not convex, and we employ MM to solve  it. That is, we first approximate the rates by the concave lower bounds in \eqref{l-r-lk}-\eqref{l-r-lk-bar}, which results in
\begin{subequations}\label{ee-opt-3}
\begin{align}
\max_{\{\mathbf{P}\}\in\mathcal{P},e}\,\,& e &
\text{s.t.}\hspace{.5cm}&\frac{\lambda_{lk}\tilde{r}_{lk}}{P_c+\eta\text{Tr}\left(\mathbf{P}_{lk}\right)}\geq  e,&\forall l,k
%\\
%&&&\frac{r_{l\bar{k}}^l}{P_c+\eta\text{Tr}\left(\mathbf{P}_{l\bar{k}}\right)}\geq \alpha_{l\bar{k}} e,&\forall l,\bar{k}
\\
&&& \tilde{r}_{lk}\geq r^{th}_{lk}, &\forall l,k.
%\label{qos-55-b}
%\\&&& r_{l\bar{k}}^l\geq r^{th}_{l\bar{k}}, &\forall l,\bar{k}.
%\label{qos-55-c}
\end{align}
\end{subequations}
Although \eqref{ee-opt-3} is not convex, we can obtain its global optimal solution by the GDA. 
The GDA is an extension of  the Dinkelbach algorithm to solve multiple-ratio fractional programming  problems \cite{zappone2015energy}. The optimum solution of  \eqref{ee-opt-3} is obtained by iteratively solving
\begin{subequations}\label{ee-opt-4}
\begin{align}
\max_{\{\mathbf{P}\}\in\mathcal{P},e}\,\,& e %&
\\
\text{s.t.}\hspace{.5cm}&
\lambda_{lk}\tilde{r}_{lk}-
\mu^{(t,n)}\left(P_c+\eta\text{Tr}\left(\mathbf{P}_{lk}\right)\right)\geq e,
&\forall l,k,
%\\
%&r_{l\bar{k}}^l-
%\mu^{(t,n)}\left(P_c+\eta\text{Tr}\left(\mathbf{P}_{l\bar{k}}\right)
%\right)
%\geq \alpha_{l\bar{k}} e,&
%\forall l,\bar{k},
\\
%&&
& \tilde{r}_{lk}\geq r^{th}_{lk}, &\forall l,k,
%\label{qos-55-b}
%\\
%& r_{l\bar{k}}^l\geq r^{th}_{l\bar{k}}, &\forall l,\bar{k}.
%\label{qos-55-c}
\end{align}
\end{subequations}
and updating $\mu^{(t,n)}$ as
\begin{align}
\nonumber
\mu^{(t,n)}=&
\min\left\{
\underset{\forall l,k}{\min}
\left\{
\frac{\lambda_{lk}\tilde{r}_{lk}\left(\{\mathbf{P}^{(t,n-1)}\},\{\bm{\Theta}^{(t-1)}\}\right)}{
%\lambda_{lk}\left(
P_c+\eta\text{Tr}\left(\mathbf{P}_{lk}^{(t,n-1)}\right)
%\right)
}
\right\},
\right.
\\&
\left.
\underset{\forall l,\bar{k}}{\min}
\left\{
\frac{\lambda_{lk}r_{l\bar{k}}^l\left(\{\mathbf{P}^{(t,n-1)}\},\{\bm{\Theta}^{(t-1)}\}\right)}{
%\lambda_{l\bar{k}}\left(
P_c+\eta\text{Tr}\left(\mathbf{P}_{l\bar{k}}^{(t,n-1)}\right)
%\right)
}
\right\}
\right\}.
\end{align}
The GDA linearly converges to the optimum solution of  \eqref{ee-opt-3}.

\begin{figure*}[t!]
    \centering
\includegraphics[width=.9\textwidth]{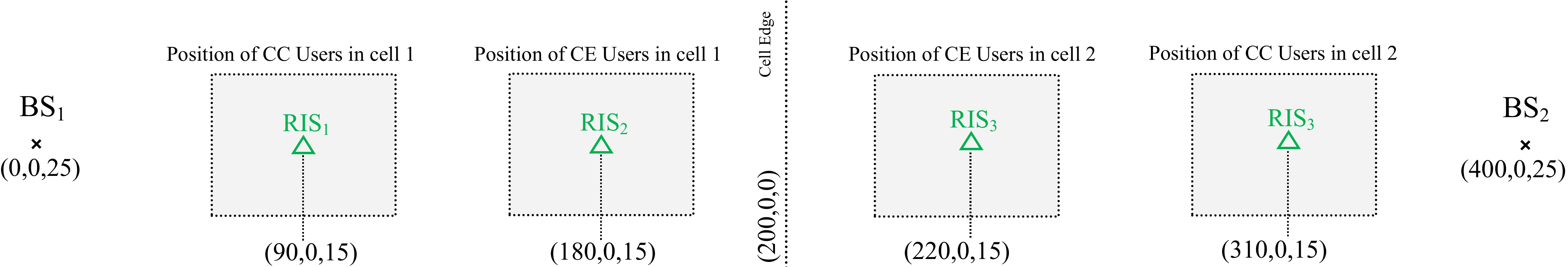}
     \caption{System topology.}
	\label{Fig-sys-model-3}
\end{figure*}

\subsection{Optimization of the reflecting coefficients}

In this subsection, we consider the optimization of the weighted-minimum-EE over reflecting coefficients for a fixed set of covariance matrices. 
The MWEEM problem for a fixed $\{\mathbf{P}^{(t)}\}$ can be simplified to 
%\begin{subequations}\label{ee-opt-ref}
\begin{align}\label{ee-opt-ref}
\max_{\{\bm{\Theta}\}\in\mathcal{T},e}\,\,& e &
\text{s.t.}\hspace{.5cm}&
r_{lk}
\geq 
\min\left(e\tilde{\lambda}_{lk},r^{th}_{lk}\right),&\forall l,k,
%\\
%&&&r_{l\bar{k}}\geq \tilde{\alpha}_{l\bar{k}} e,&\forall l,\bar{k},
%\\&&& r_{lk}\geq r^{th}_{lk}, &\forall l,k,
%\label{qos-55-b}
%\\&&& r_{l\bar{k}}^l\geq r^{th}_{l\bar{k}}, &\forall l,\bar{k},
%\label{qos-55-c}
\end{align}
%\end{subequations}
where $\tilde{\lambda}_{lk}=\left(
P_c+\eta\text{Tr}\left(\mathbf{P}_{lk}^{(t)}\right)
\right)/\lambda_{lk}$. %and $\tilde{\alpha}_{l\bar{k}}=\alpha_{l\bar{k}}\left(P_c+\eta\text{Tr}\left(\mathbf{P}_{l\bar{k}}^{(t)}\right)\right)$.
The structure of this problem is similar to the MWRM problem in Section \ref{sec-orm}. 
Hence, we can easily employ the algorithms in Section \ref{sec-orm} to solve \eqref{ee-opt-ref}. 
To this end, we first substitute the rates by the lower bounds in Theorem \ref{theo-1}, which yields
%\begin{subequations}\label{ee-opt-ref-2}
\begin{align}\label{ee-opt-ref-2}
\max_{\{\bm{\Theta}\}\in\mathcal{T},e}\,\,& e &
\text{s.t.}\hspace{.5cm}&
\hat{r}_{lk}
\geq 
\min\left(\tilde{\lambda}_{lk} e,r^{th}_{lk}\right),&\forall l,k.
\end{align}
%\end{subequations}
This optimization problem is very similar for the different feasibility sets and can be solved  similarly, so we skip the details due to limited space. We summarize our proposed algorithm for the MWEEM in Algorithm II.
\begin{table}[htb]
\label{alg-ii}
\begin{tabular}{l}
%\rule{0.55\textwidth}{2pt}\par\\
%\rule{\hsize}{2pt}\par
\hline 
 \small{\textbf{Algorithm II} NOMA-based IGS algorithm in Sec. \ref{EE-sec} with $\mathcal{T}$.}\\
%\rule{0.55\textwidth}{1.5pt}\par\\
\hline 
\hspace{0.2cm}\small{\textbf{Initialization}}\\
\hspace{0.2cm}\small{Set $\epsilon$, %$L$,  
$t=1$, %convergence=0, 
 $\{\mathbf{P}\}=\{\mathbf{P}^{(0)}\}$, and$\{\bm{\Theta}\}=\{\bm{\Theta}^{(0)}\}$ }\\
%\hspace{0.2cm}\small{{\bf For} $k=1,\cdots,K$}\\
%\hspace{1.2cm}\small{$p_k=P_k$}\\
%\hspace{0.2cm}\small{\textbf{End (For)}}\\
\hline 
\hspace{0.2cm}%\small{\textbf{Repeat}}\\
\small{\textbf{While} $\left(\!\!\underset{\forall l,k}{\min}\,\lambda_{lk}EE^{(t)}_{lk}\!-\underset{\forall l,k}{\min}\,\lambda_{lk}EE^{(t-1)}_{lk}\!\!\right)\!\!/\underset{\forall l,k}{\min}\,\lambda_{lk}EE^{(t-1)}_{lk}\!\!\geq\!\epsilon$ }\\ %and $t\leq T$ \textbf{do}
\hspace{.6cm}\small{{\bf Optimizing over} $\{\mathbf{P}\}$ {\bf by fixing} $\{\bm{\Theta}^{(t-1)}\}$}\\
\hspace{1.2cm}\small{Obtain $\tilde{r}_{lk}^{(t-1)}$ %\left(\{\mathbf{P}\},\{\bm{\Theta}^{(t-1)}\}\right)$ 
based on \eqref{l-r-lk}-\eqref{l-r-lk-bar}}\\ % and \eqref{l-r-lk-bar-low-32}}\\
\hspace{1.2cm}\small{Compute $\{\mathbf{P}^{(t)}\}$ by solving \eqref{ee-opt-3}}\\
\hspace{.6cm}\small{{\bf Optimizing over} $\{\bm{\Theta}\}$ {\bf by fixing} $\{\mathbf{P}^{(t-1)}\}$}\\
\hspace{1.2cm}\small{Obtain $\hat{r}_{lk}^{(t-1)}$ %\left(\{\mathbf{P}^{(t)}\},\{\bm{\Theta}\}\right)$ 
based on Corollary 1}\\
\hspace{1.2cm}\small{Compute $\{\bm{\Theta}^{(t)}\}$ by solving \eqref{ee-opt-ref-2}}\\
%\hspace{1.2cm}\small{\textbf{End (For)}}\\
%\hspace{1.2cm}\small{{\bf If} $\left(\sum_{k=1}^{K}R^{(l+1)}_k-\sum_{k=1}^{K}R^{(l)}_k\right)/\sum_{k=1}^{K}R^{(l)}_k<\epsilon$}\\
%\hspace{1.8cm}\small{convergence=1}\\
%\hspace{1.2cm}\small{\textbf{End (If)}}\\
\hspace{.6cm}\small{$t=t+1$}\\
\hspace{0.2cm}\small{\textbf{End (While)}}\\
\hspace{0.2cm}\small{{\bf Return} $\{\mathbf{P}^{\star}\}$ and $\{\bm{\Theta}^{\star}\}$.}\\
\hline %\rule{0.55\textwidth}{2pt}\par
\end{tabular}
\end{table}

\subsection{Discussion on computational complexity}
The computation complexity of our proposed algorithm for WMEEM can be computed similar to the approach in Section \ref{sec=com=}. The computational complexities involved with updating the RIS components is similar to  Section \ref{sec=com=} since the optimization problems \eqref{opt-prob-wmrm-rc-u} and \eqref{ee-opt-ref-2} have a similar structure, and thus, their computational complexities are in the same order. 
However, updating the transmit covariance matrices has higher computational complexities comparing to the WMRM problem since \eqref{ee-opt-3} has a fractional structure, and we employ the GDA to solve it.  Indeed, we have to solve a sequence of convex optimization problems (\eqref{ee-opt-4}) to update the transmit covariance matrices. 
We set the maximum number of the GDA iterations $N_I$. 
The maximum number of Newton iterations to solve \eqref{ee-opt-4} is related to $\sqrt{4LK}$ since the number of constraints in \eqref{ee-opt-4} is $4LK$. The main complexity in solving each Newton iteration is to compute the lower bounds for the user rates, which approximately has $8+24LN_{BS}^2$ multiplications. Thus, the overall computational complexity of solving the WMEEM problem can be approximated as
\begin{multline*}
 \mathcal{O}\left(N\left(N_I\sqrt{KL}\left(L^2KN_{BS}^2\right)\right.\right.\\
 +\left.\left. LK\sqrt{KL+MN_{RIS}}\left(N^2_{BS}L+MN_{RIS}N_{BS}\right)\right)\right).
 \end{multline*}

\section{Numerical Results}\label{sec-v-nr}
%In this section, we present some numerical results.
 We consider a two-cell broadcast channel with $2K$ users as well as two RIS locations in each cell  as shown in Fig. \ref{Fig-sys-model-3}, unless it is explicitly mentioned otherwise. The dimensions in Fig. \ref{Fig-sys-model-3} are in meters.  We assume that the users are uniformly located into two different clusters/positions in each cell: a cell centric and a cell edge region \cite{tuan2019non}. In each cluster, the users are located in a $20$m$\times 20$m square with a RIS at its center. Note that the position of each user is modeled by a random variable with a uniform distribution and varies in each iteration of the Monte Carlo simulations.   
The heights of BSs, RISs and users are, respectively, 25m, 15m and 1.5m. 
It is worth emphasizing that we consider a specific system topology to accurately model the large-scale fading, which may have a great impact on the performance of RIS \cite{huang2019reconfigurable, kammoun2020asymptotic, pan2020multicell, yu2020joint}. 
We further assume that the power budget of each BS is equal to $P$ watts. Additionally, we assume that the weights of all users are equal to 1, i.e., $\lambda_{l,k}=1$ for all $l,k$.
The parameters for the feasibility set $\mathcal{T}_{C}$ are $|\theta|_{\min}=0.2$, $\alpha=1.6$, and $\phi=0.43\pi$ \cite{abeywickrama2020intelligent}.
We employ a random feasible initial point for running the proposed algorithms. 
We assume that the direct link between a BS and a user is non-line-of-sight (NLoS), which implies that the small-scale fading for a direct link is Rayleigh. It means that the channel coefficients for direct links are drawn from a zero-mean complex proper Gaussian distribution with unit variance times an attenuation coefficient, which models the large-scale fading. %We will discuss the large-scale fading later in this subsection.
Furthermore, we assume that the links related to RISs are line-of-sight (LoS), which implies that their small-scale fading is Rician. Thus, the direct link between BS $i$ and RIS $m$, denoted by $\mathbf{G}_{mi}$, is modeled as
\cite{pan2020multicell}
\begin{equation*}
\mathbf{G}_{mi}=\sqrt{\beta_{mi}}\left(\sqrt{\frac{\gamma}{1+\gamma}}\mathbf{G}_{mi}^{\text{LoS}}+\sqrt{\frac{1}{1+\gamma}}\mathbf{G}_{mi}^{\text{NLoS}}\right),
\end{equation*}
where $\beta_{mi}$ is the channel-attenuation coefficient due to the large-scale fading,  $\gamma$ is the Rician factor, $\mathbf{G}^{\text{LoS}}_{mi}$ is the LoS component, and $\mathbf{G}^{\text{NLoS}}_{mi}$ is the NLoS Rayleigh component. We consider $\gamma=3$ for all links.
Although the NLoS component $\mathbf{G}^{\text{NLoS}}_{mi}$ is a random variable, the LoS component $\mathbf{G}^{\text{LoS}}_{mi}$ is a deterministic variable modeled as \cite{pan2020multicell}
\begin{equation*}
\mathbf{G}^{\text{LoS}}_{mi}=\mathbf{a}_{r_{mi}}\left(\phi^{A}_{mi}\right)\mathbf{a}^H_{t_{mi}}\left(\phi^{D}_{mi}\right),
\end{equation*} 
where $\phi^{A}_{mi}/\phi^{D}_{mi}\sim\text{Unif}[0,2\pi]$ is the angle of arrival/departure for the link between BS $i$ and RIS $m$, and \cite{pan2020multicell}
\begin{align*} 
\mathbf{a}_{r_{mi}}
%\left(\phi^{AoA}_{mi}\right)
&=\left[1,e^{j\frac{2\pi d\sin(\phi^{A}_{mi})}{\lambda}},\cdots,e^{j\frac{2(N_{RIS}-1)\pi d\sin(\phi^{A}_{mi})}{\lambda}}\right],
%\in\mathbb{C}^{N_{RIS}\times 1},
\\
%&
\mathbf{a}_{t_{mi}}%\left(\phi^{AoD}_{mi}\right)
&=\left[1,e^{j\frac{2\pi d\sin(\phi^{D}_{mi})}{\lambda}},\cdots,e^{j\frac{2(N_{BS}-1)\pi d\sin(\phi^{D}_{mi})}{\lambda}}\right].%\in\mathbb{C}^{N_{BS}\times 1}.
\end{align*}
 Note that the difference between LoS and NLoS links is not only in the   small-scale fading, but also  they have different large-scale fading.  
The path loss in dB for each link can be modeled as \cite{yu2020joint, soleymani2022improper}
\begin{equation*}
\text{PL}=\text{PL}_0+G_0-10\alpha\log_{10}\left(\frac{d}{d_0}\right),
\end{equation*}
where $\text{PL}_0$ is the path loss at the reference distance $d_0$, $G_0$ is the transmitter antenna gain, $\alpha$ is the path-loss exponent, and $d$ is the distance between the transmitter and receiver. 
The effective channel attenuation coefficient is $\beta=10^{\text{PL}/10}$ \cite{ soleymani2022improper}. Note that the path-loss exponent $\alpha$ depends to a large extent on whether a link is LoS or NLoS, where NLoS links experience a larger attenuation.
The simulation parameters related to the large-scale and small-scale fading are the same as in \cite{soleymani2022improper}. Due to the space restriction, we refer the readers to \cite[Table II]{soleymani2022improper} for additional details on the rest of the parameters. Employing Monte Carlo simulations, we average the results over 100 channel realizations. 

In this paper, we compare the proposed scheme with the algorithms in \cite{soleymani2022improper}, which considers TIN. We also compare the proposed algorithms with the traditional systems, i.e., systems with no RIS. To summarize, the considered algorithms in the simulations are as follows:
\begin{itemize}
\item 
{\bf IR$_X$N} (or {\bf PR$_X$N}) refers to the proposed NOMA-based IGS (or PGS) scheme for $\mathcal{T}_X$, where X can be equal to U, I and C. 
\item 
{\bf IR$_R$N} (or {\bf PR$_R$N}) refers to the NOMA-based IGS (or PGS) scheme with random reflecting coefficients.

\item 
{\bf IR$_U$T} (or {\bf PR$_U$T}) refers to the proposed IGS (or PGS) scheme in \cite{soleymani2022improper} for the feasibility set $\mathcal{T}_U$ with TIN. Note that this scheme can be seen as an upper bound for the IGS (or PGS) performance with TIN.
\item 
{\bf IN} (or {\bf PN}) refers to the proposed NOMA-based IGS (or PGS) scheme, which is applied to traditional multicell BCs (i.e., without RIS).
\end{itemize}

\begin{figure*}[t!]
    \centering
    \begin{subfigure}[t]{0.32\textwidth}
        \centering
        \includegraphics[width=.92\textwidth]{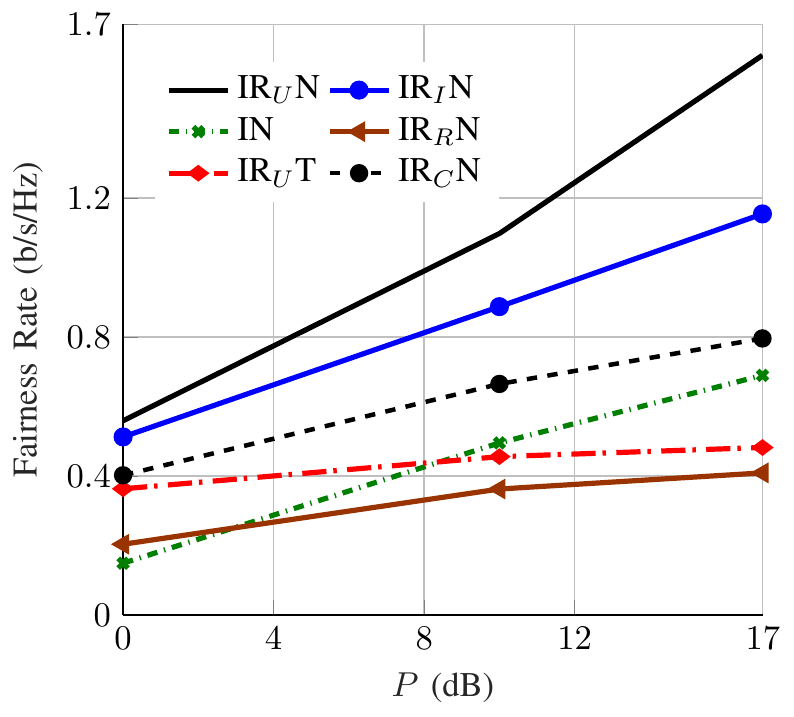}
        \caption{IGS schemes.}
    \end{subfigure}%
    ~ 
    \begin{subfigure}[t]{0.32\textwidth}
        \centering
\includegraphics[width=.92\textwidth]{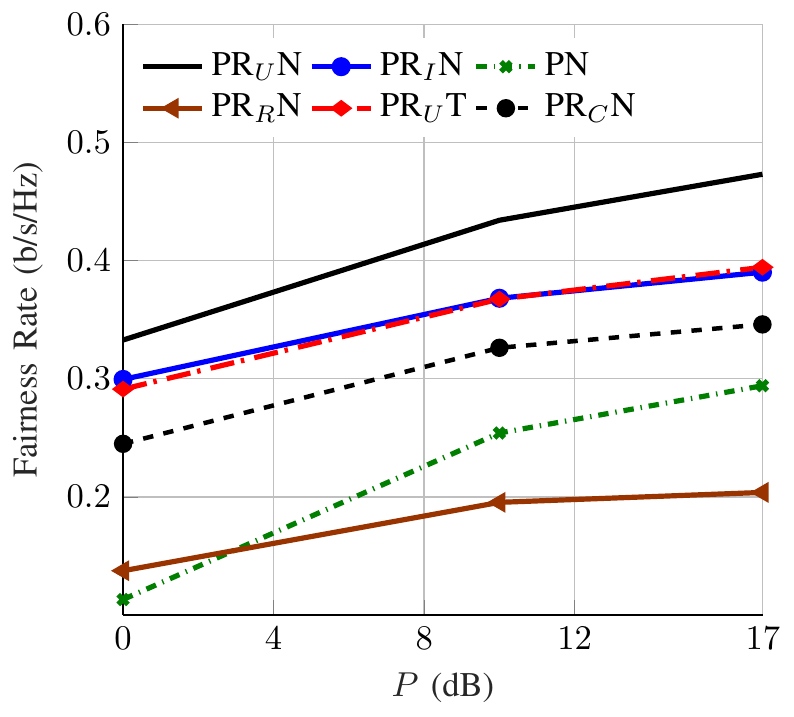}
        \caption{PGS schemes.}
    \end{subfigure}
	~
	\begin{subfigure}[t]{0.32\textwidth}
        \centering
\includegraphics[width=.92\textwidth]{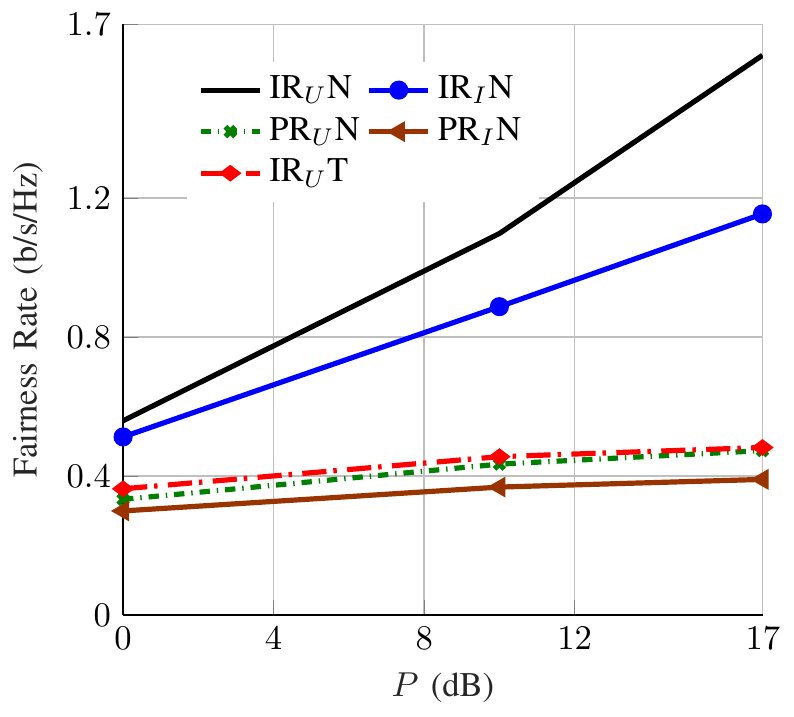}
        \caption{IGS and PGS schemes with RIS.}
    \end{subfigure}
    \caption{The average fairness rate versus $P$ for $N_{BS}=N_u=1$, $N_{RIS}=25$, $L=2$, $K=2$, $M=4$.}
	\label{Fig-2-nrs} 
\end{figure*}

\subsection{Fairness rate}
In this subsection, we consider the MWRM problem. We refer to the maximum of the minimum rate as a fairness rate since it is usually the case that all users receive an equal achievable rate when the minimum rate is maximized \cite{soleymani2022improper, park2013sinr}. To fully evaluate the system performance, we show the effect of different parameters on the performance of RIS-assisted systems. These parameters are power budget, transmit antennas at the BSs, number of users and number of RIS elements. Finally, we briefly illustrate the convergence of our proposed algorithms.
\begin{figure*}[t!]
    \centering
    \begin{subfigure}[t]{0.32\textwidth}
        \centering
       \includegraphics[width=.92\textwidth]{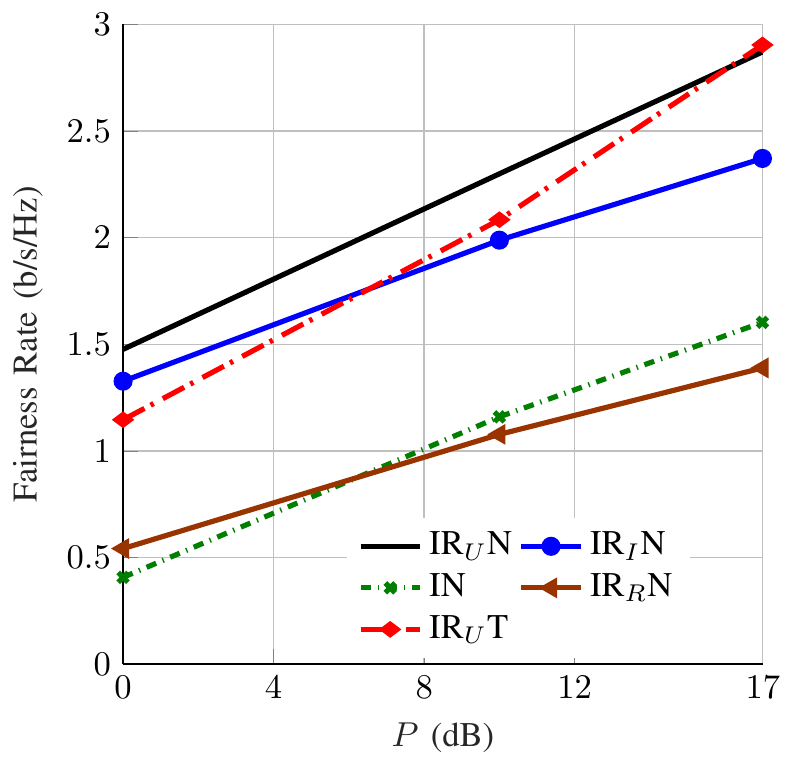}
        \caption{IGS schemes.}
    \end{subfigure}%
    ~ 
    \begin{subfigure}[t]{0.32\textwidth}
        \centering
\includegraphics[width=.92\textwidth]{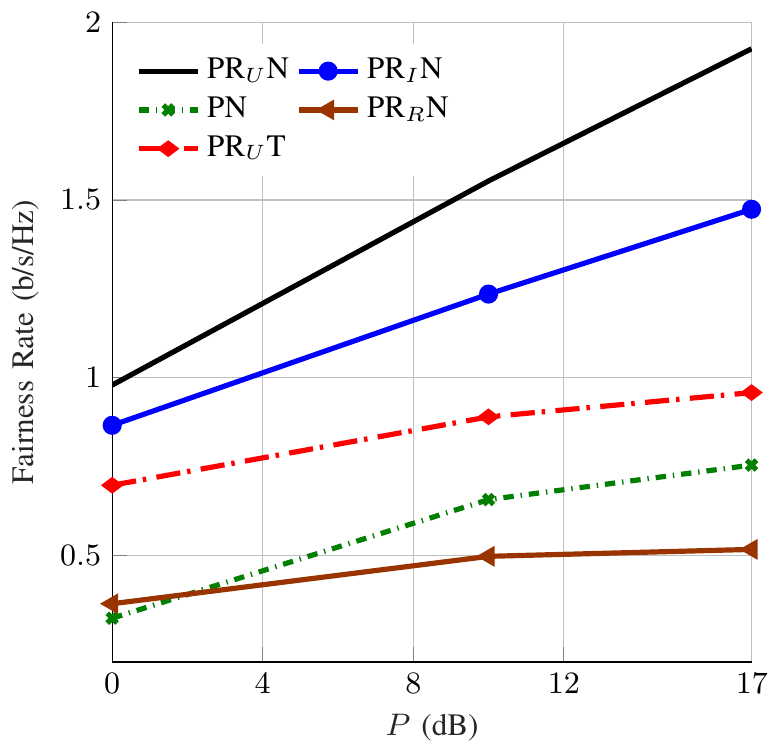}
        \caption{PGS schemes.}
    \end{subfigure}
	~
	\begin{subfigure}[t]{0.32\textwidth}
        \centering
\includegraphics[width=.92\textwidth]{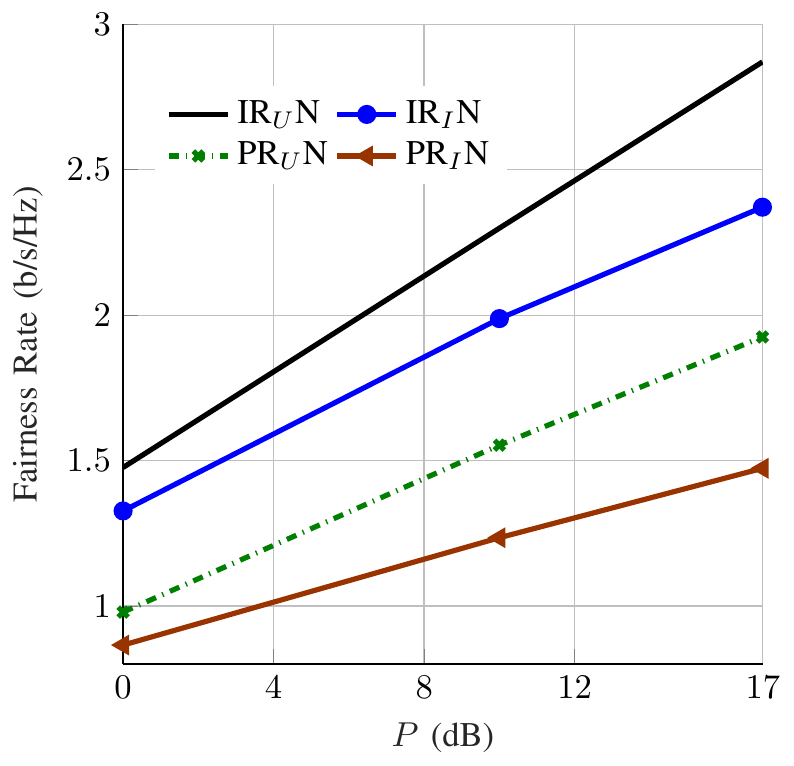}
        \caption{IGS and PGS schemes with RIS.}
    \end{subfigure}
    \caption{The average fairness rate versus $P$ for $N_{BS}=2$, $N_{RIS}=25$, $L=2$, $K=2$, $M=4$.}
	\label{Fig-3-srs} 
\end{figure*}

\subsubsection{Impact of power budget}
In Fig. \ref{Fig-2-nrs}, we show the effect of the BS  power budget on the minimum user rate. As can be observed, 
the RIS-assisted systems (both IGS and PGS) with NOMA outperform the RIS-assisted systems with TIN as well as the systems without RIS. 
However, the IGS (and PGS) schemes with random RIS phases perform even worse than the IGS (and PGS) schemes with no-RIS especially at high SNR regimes, which indicates the importance of optimizing RIS coefficients.
This is rather a surprising result, which shows that employing RIS without a proper optimization may even make the system performance worse.  
In Fig.  \ref{Fig-2-nrs}c, we compare the IGS and PGS schemes. As can be observed, the IGS scheme can significantly outperform the PGS schemes, and the benefits of IGS increase with  $P$. The reason is that the interference level increases with $P$, so more improvements are expected  by IGS as an interference-management technique. 
Another interesting behavior in Fig. \ref{Fig-2-nrs} is that IGS increases the benefits of NOMA. 
Indeed, it suggests that these interference-management techniques can be seen as a complementary of each other since NOMA handles the interference from the paired user (intracell), while IGS mainly handles ICI as well as the remaining intracell interference. Furthermore, we can observe that RIS increases the benefits of IGS.  

In Fig. \ref{Fig-2-nrs}, we also consider the impact of the feasibility set for the RIS components. As can be observed, RIS can improve the system performance with the three considered feasibility sets. Moreover, we observe that the amplitude of RIS components can have a considerable impact on the performance. As expected, RIS with the feasibility set $\mathcal{T}_U$ outperforms the other feasibility sets, which are a subset of $\mathcal{T}_U$.  

In Fig. \ref{Fig-3-srs}, we show the average fairness rate versus $P$ for $N_{BS}=2$, $N_{RIS}=25$, $L=2$, $K=2$, $M=4$. As can be observed, the overall performance is similar to the single-antenna systems. 
For instance, as shown in Fig. \ref{Fig-3-srs}c, the IGS scheme can still significantly outperform the PGS scheme. 
The main difference between Fig. \ref{Fig-3-srs} and Fig. \ref{Fig-2-nrs} is in the performance gap between the proposed NOMA scheme and the TIN scheme proposed in \cite{soleymani2022improper}, which decreases by increasing the number of transmit antennas. The reason is that by increasing the number of transmit antennas, we actually increase the number of spatial resources, which consequently,  reduces the interference level. It is expected that the benefits of interference-management techniques decreases whenever there is a low interference level, as will be corroborated in the next example. 
  
In Fig. \ref{Fig-3-srs}, we also observe an interesting behavior in the performance of NOMA. That is, IGS with TIN performs very close to IGS with NOMA at high SNR. This happens since the channels in IGS are not degraded, which means that NOMA can be suboptimal especially at high SNRs. This in turn implies that we should carefully choose the signaling scheme based on the operational regime. We elaborate more on this issue in this section.

\subsubsection{Impact of transmit antennas}
\begin{figure}[t!]
    \centering
\includegraphics[width=.35\textwidth]{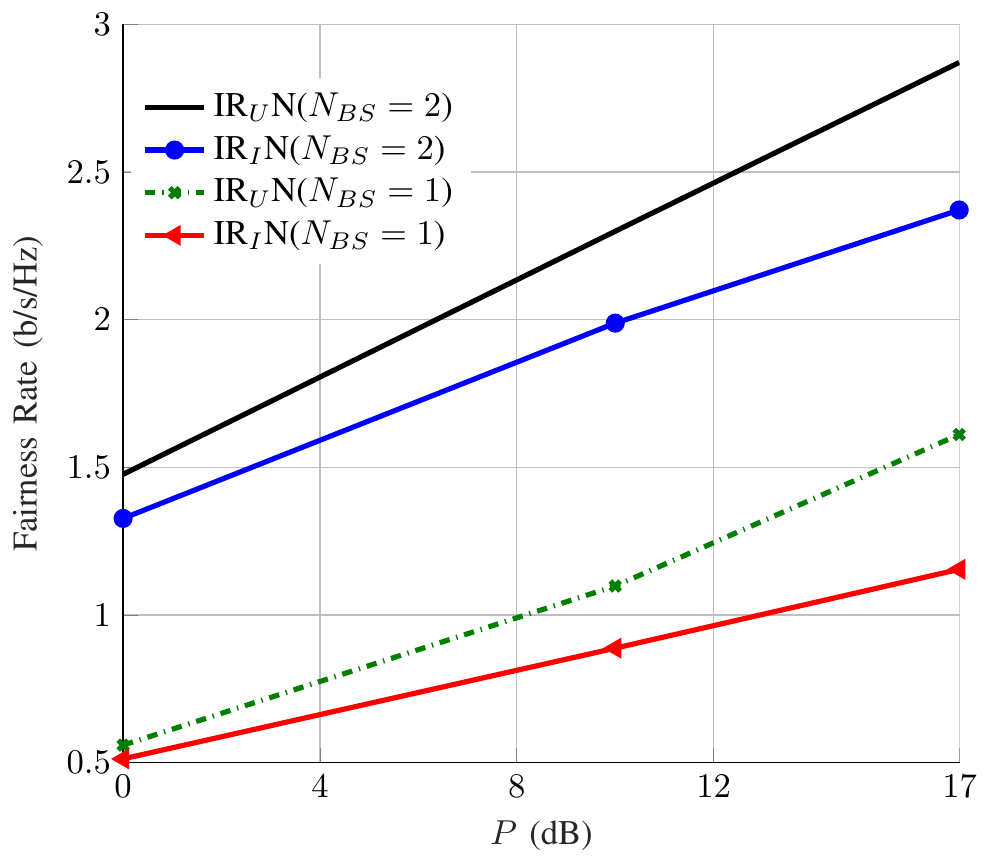}
    \caption{The average fairness rate versus $P$ for $N_u=1$, $N_{RIS}=25$, $L=2$, $K=2$, $M=4$.}
	\label{Fig-8-nrs} 
\end{figure}
Fig. \ref{Fig-8-nrs} shows the average fairness rate for IGS schemes versus $P$ for $N_{RIS}=25$, $L=2$, $K=2$, $M=4$ and different number of transmit antennas at the BSs.  
As can be observed, the average fairness rate significantly increases with the number of transmit antennas. Moreover, we observe a similar behavior for the IGS schemes with different feasibility sets. Hence, it can be expected that the performance gap between the upper bound performance of RIS remains relatively constant if we increase the number of transmit antennas. Furthermore, this gap increases with the power budget (or equivalently the transmission power) of the BS.

\begin{figure}[t!]
    \centering
\includegraphics[width=.35\textwidth]{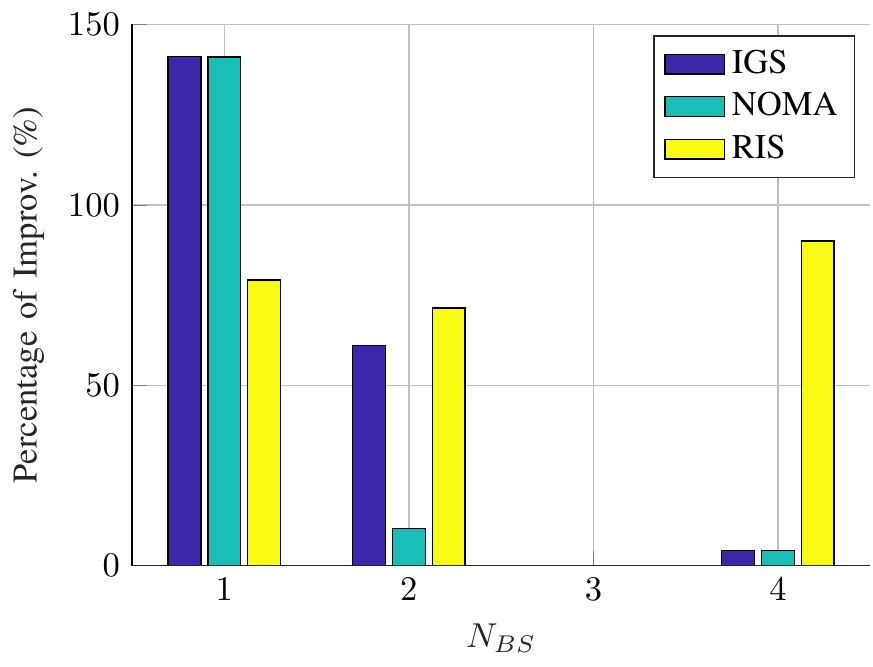}%{fig/rr-n-bs-ip-ben}
    \caption{The performance improvement by IGS/NOMA/RIS versus $N_{BS}$ for $P=10$W, $N_{RIS}=25$, $L=2$, $K=2$, $M=4$.}
	\label{Fig-nbs-nrs} 
\end{figure}
In Fig. \ref{Fig-nbs-nrs}, we show the relative performance improvements in the average fairness rate by each technique  versus $N_{BS}$ for $P=10$W, $N_{RIS}=25$, $L=2$, $K=2$, $M=4$. 
We obtain the relative improvements of IGS by comparing the average fairness rate of IR$_I$N and PR$_I$N schemes. 
Furthermore, the benefits of NOMA are obtained by comparing the performance of IR$_I$N and IR$_I$T schemes. Finally, the RIS benefits are evaluated by comparing the performance of IR$_I$N and IN schemes. 
As can be observed, the benefits of IGS as well as NOMA decrease with the number of transmit antennas. 
The reason is that the number of (spatial) resources per users increases with the number of antennas, which in turn decreases the interference level. Thus, the benefits of IGS as an interference-management technique are expected to be reduced. When the number of transmit antennas is lower than the number of users per cell, $N_{BS}< 2K$, BSs are overloaded, and the interference level is very high. 
As a result, we should employ interference-management techniques such as IGS and/or NOMA. 
However, when $N_{BS}> 2K$, BSs are underloaded, and the interference-level is low. Thus, we can simply treat the interference as noise and employ PGS without a considerable performance degradation. We also observe in Fig. \ref{Fig-nbs-nrs} that the benefits of RIS remain relatively constant when we switch to underloaded systems. This interesting result can show that the main role of RIS in this set up is to improve the coverage and channel condition of users. Thus, RIS can significantly improve the system performance in both underloaded and overloaded regimes. However, in overloaded regimes, we have to employ interference-management techniques to reap the benefits of RIS and achieve a better performance.

\subsubsection{Impact of number of users per cell}
%\begin{figure}[t!]
%    \centering
%       \input{rr-k-ip}
%    \caption{The average fairness rate versus $K$ for $N_{BS}=2$, $N_{RIS}=25$, $L=2$,  $M=2$.}
%	\label{Fig-10-nrs} 
%\end{figure}
\begin{figure}[t!]
    \centering
    \begin{subfigure}[t]{0.24\textwidth}
        \centering
\includegraphics[width=\textwidth]{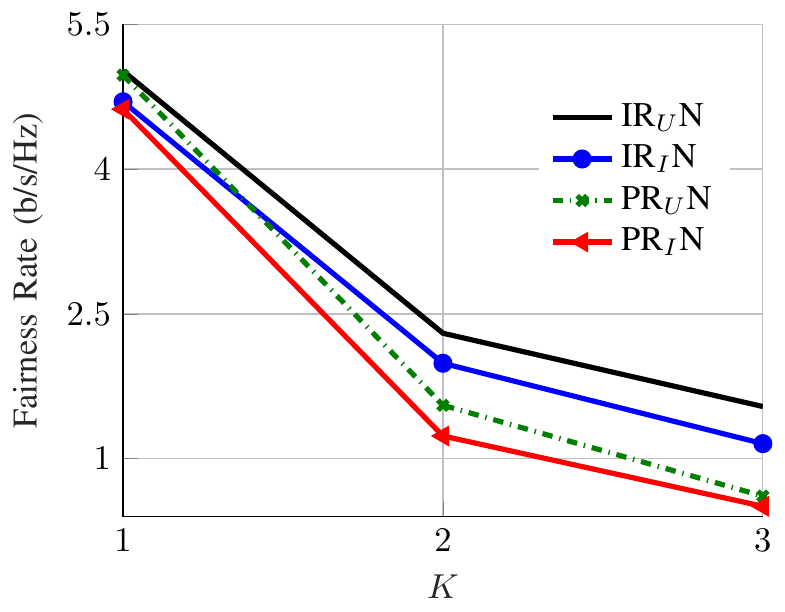}
        \caption{Average fairness rate.}
    \end{subfigure}%
    ~ 
    \begin{subfigure}[t]{0.24\textwidth}
        \centering
\includegraphics[width=.95\textwidth]{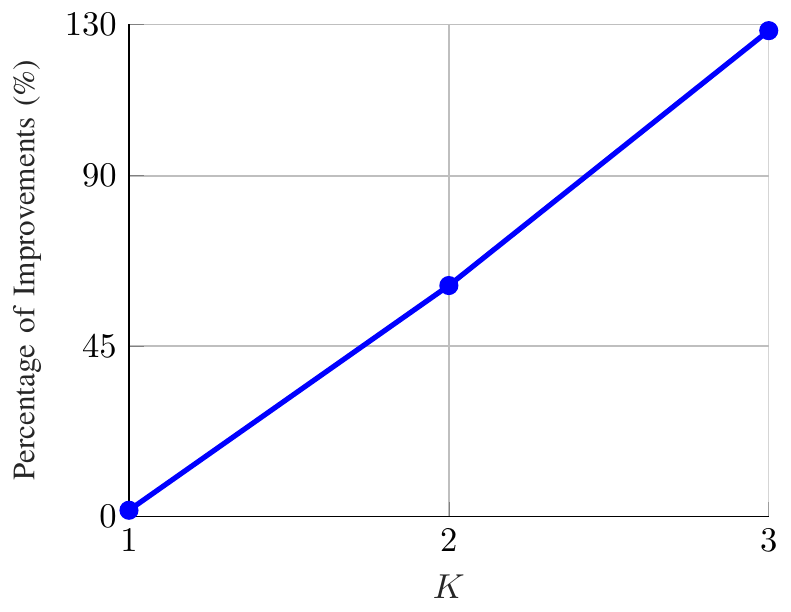}
        \caption{Benefits by IGS.}
    \end{subfigure}
	 \caption{The average fairness rate versus $K$ for $N_{BS}=2$, $P=10$, $N_{RIS}=25$, $L=2$,  $M=4$.}
	\label{Fig-10-nrs} 
\end{figure}
Fig. \ref{Fig-10-nrs} shows the average fairness rate versus $K$ for $N_{BS}=2$, $N_{RIS}=25$, $L=2$,  $M=4$. 
As can be observed, the benefits of employing IGS increase with the number of users. 
Indeed, there are only minor benefits by IGS when the number of users per cell is low. Particularly, when $K=1$, there is no intracell interference for the CCUs, and hence, the interference level is significantly reduced, which in turn, considerably decreases the benefits of IGS.  
On the contrary, the interference level increases with the number of users per cell for a fixed number of antennas. 
Thus, the benefits of IGS significantly increase with the number of users per cell as illustrated in Fig. \ref{Fig-10-nrs}b. 
Through Figs. \ref{Fig-nbs-nrs} and \ref{Fig-10-nrs}, we observe that we can get a very high gain by interference-management techniques whenever the BSs are overloaded, i.e., when the number of users per cell is larger than the number of transmit antennas $(2K>N_{BS})$. 
However, we can employ a simple PGS scheme with TIN when BSs are underloaded $(2K\leq N_{BS})$ without a considerable performance loss.

\subsubsection{Impact of number of RIS components}
\begin{figure}[t!]
    \centering
\includegraphics[width=.45\textwidth]{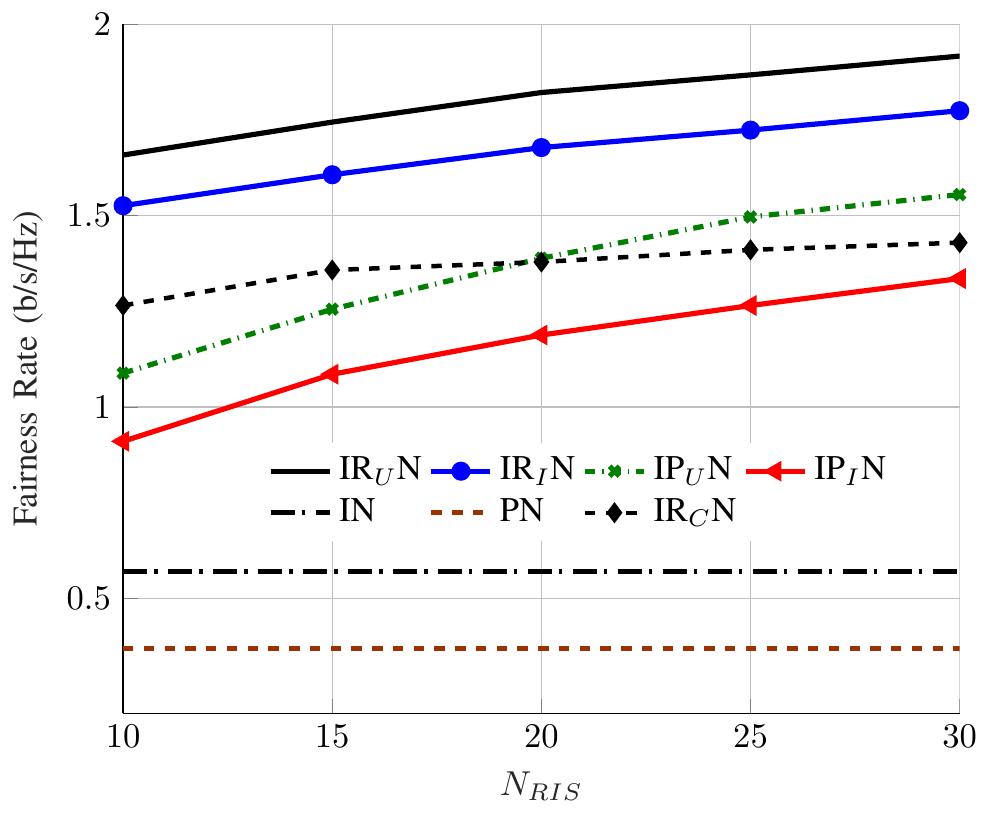}
    \caption{The average fairness rate versus $N_{RIS}$ for $N_{BS}=2$, $P=10$,  $M=2$.}
	\label{Fig-11-nrs} 
\end{figure}
Figure \ref{Fig-11-nrs} shows the average fairness rate versus $N_{RIS}$ for $N_{BS}=2$, $P=10$,  $M=2$. In this figure, there is only one RIS per cell, which is located close to the CEUs. 
As expected, the average minimum rate increases with $N_{RIS}$; however, the improvements in the rate are not significant, and even with a relatively low number of RIS components ($N_{RIS}=10$), RIS can provide a considerable gain.   
This means that we do not necessarily need a high number of RIS components to get benefits, which can be interesting from practical point of view.

\subsubsection{Convergence}
\begin{figure}[t!]
    \centering
    \begin{subfigure}[t]{0.24\textwidth}
        \centering
\includegraphics[width=\textwidth]{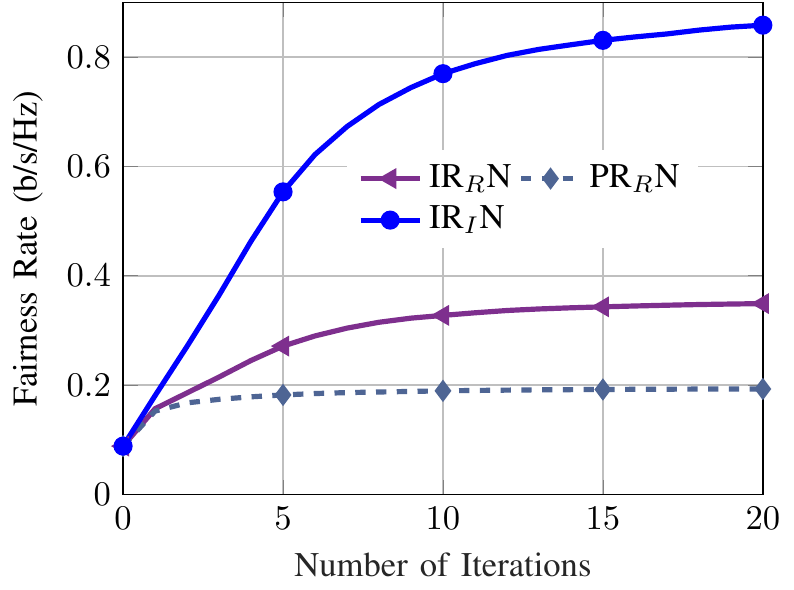}
        \caption{$N_{BS}=1$ and $P=1$W.}
    \end{subfigure}%
    ~%\\ 
    \begin{subfigure}[t]{0.24\textwidth}
        \centering
\includegraphics[width=\textwidth]{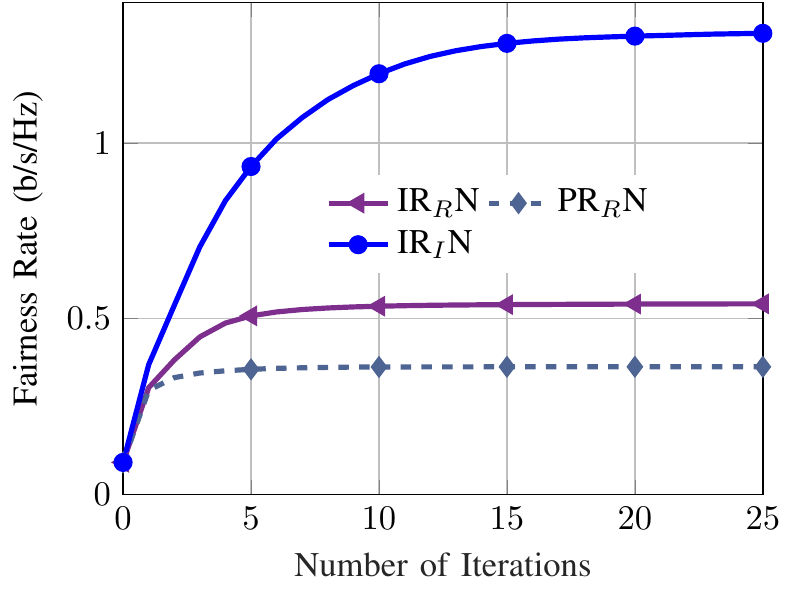}
        \caption{$N_{BS}=2$ and $P=1$W.}
    \end{subfigure}
    \caption{The average fairness rate versus number of iterations for $N_{RIS}=25$, $L=2$, $K=2$, and $M=4$.}
	\label{Fig-4} 
\end{figure}
Figure \ref{Fig-4} shows the average fairness rate versus the number of iterations. 
As expected, the algorithms for RIS-assisted systems require more iterations for convergence and are slower than the algorithms that do not optimize over RIS components. 
However, the RIS-assisted algorithms outperform the other algorithms after only a few iterations. 
For instance, in Fig. \ref{Fig-4}b, the IGS algorithm provides a better solution than the final value of the I-R algorithm after just 3 iterations.

\subsection{EE maximization}
\begin{figure*}[t!]
    \centering
    \begin{subfigure}[t]{0.32\textwidth}
        \centering
\includegraphics[width=\textwidth]{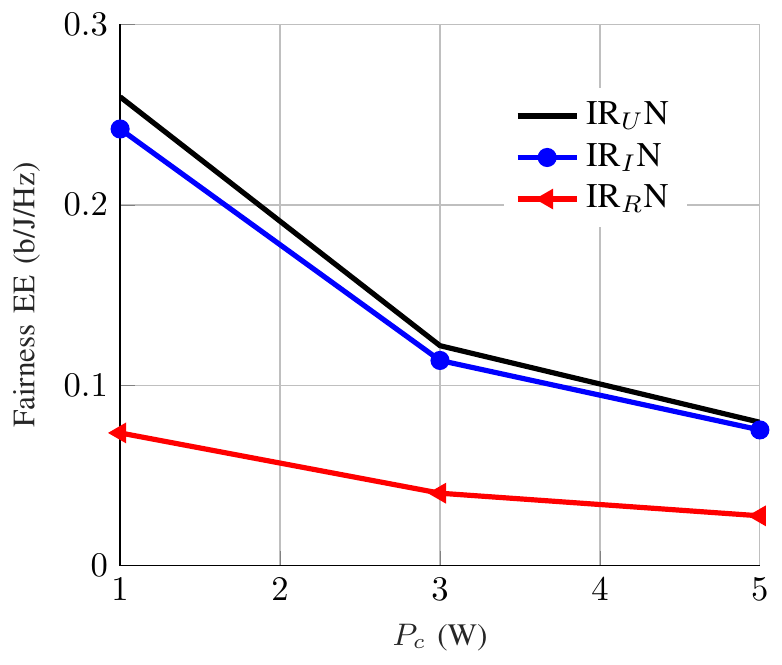}
        \caption{IGS schemes.}
    \end{subfigure}%
    ~ 
    \begin{subfigure}[t]{0.32\textwidth}
        \centering
\includegraphics[width=\textwidth]{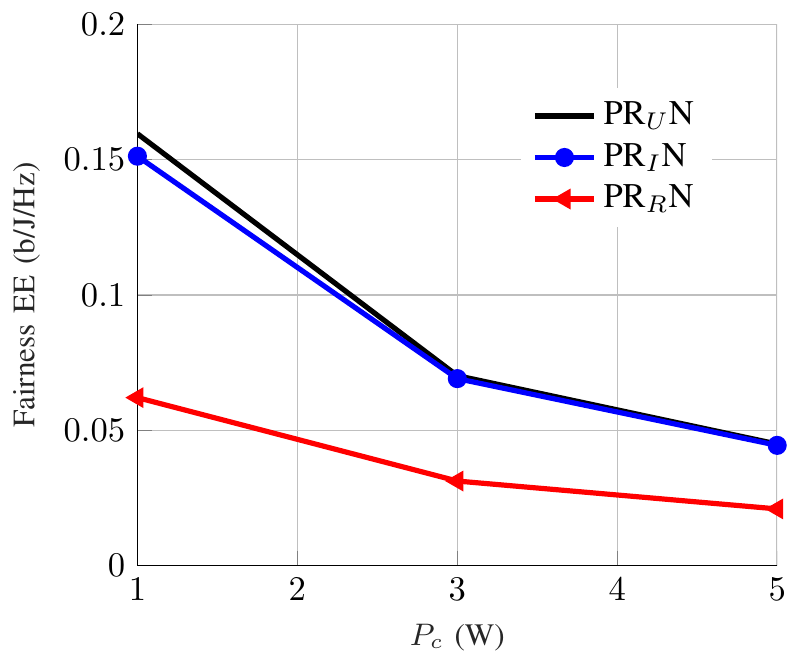}
        \caption{PGS schemes.}
    \end{subfigure}
	~
	\begin{subfigure}[t]{0.32\textwidth}
        \centering
\includegraphics[width=\textwidth]{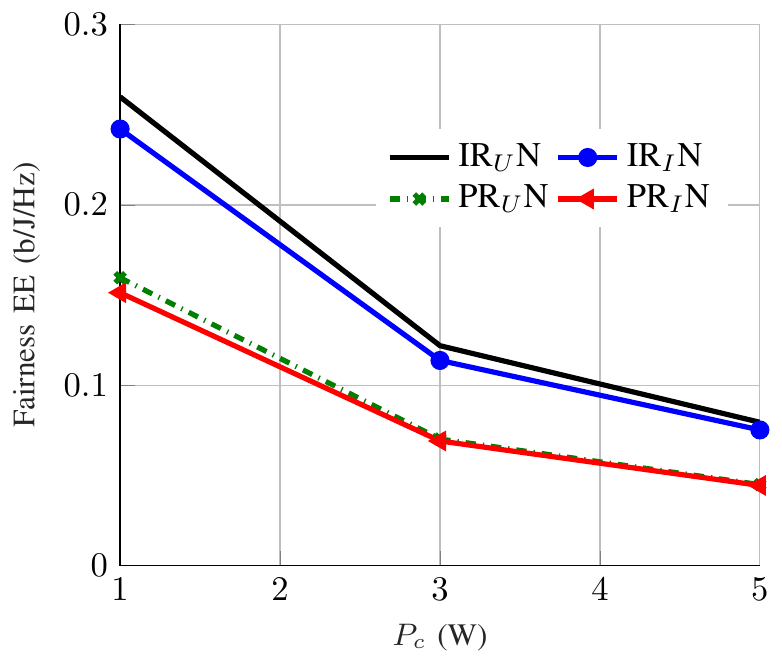}
        \caption{IGS and PGS schemes with RIS.}
    \end{subfigure}
    \caption{The average fairness EE versus $P_c$ for $P=1$w, $N_{BS}=N_u=1$, $N_{RIS}=25$, $L=2$, $K=2$, $M=2$.}
	\label{Fig-ee} 
\end{figure*}
In this subsection, we provide the results for the maximization of the minimum EE of users. 
We expect that we observe the same behavior in the performance of the proposed algorithms by varying different parameters such as $K$, $N_{BS}$, $N_{RIS}$. Hence, due to the space restriction, we provide only a figure for the EE maximization. 

In Fig. \ref{Fig-ee}, we show the minimum EE of users versus $P_c$ for $P=1$w, $N_{BS}=N_u=1$, $N_{RIS}=25$, $L=2$, $K=2$, $M=2$. As can be observed, IGS can significantly outperform PGS. 
Moreover, the difference between the IGS scheme and the upper bound (I-M scheme) is not significant, which is because of a lower transmission power. 

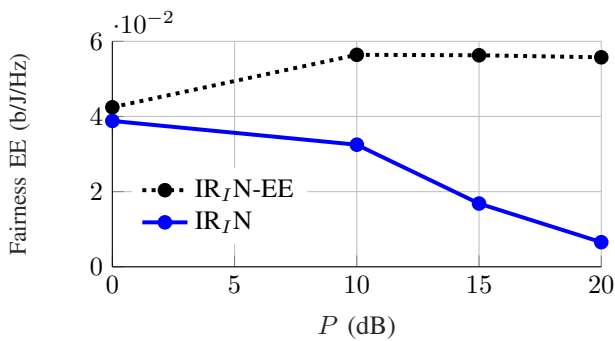
\begin{figure}[t!]
    \centering
\definecolor{mycolor1}{rgb}{0.00000,0.44700,0.74100}%
\definecolor{mycolor2}{rgb}{0.85000,0.32500,0.09800}%
\begin{tikzpicture}

\begin{axis}[%
width=6.5cm,
height=3cm,
at={(2.11in,0.971in)},
scale only axis,
xmin=0,
xmax=20,
xlabel style={font=\color{white!15!black}},
xlabel={$P$ (dB)},
ymin=0,
ymax=0.06,
ylabel style={font=\color{white!15!black}},
ylabel={\small Fairness EE (b/J/Hz)},
axis background/.style={fill=white},
axis x line*=bottom,
axis y line*=left,
xmajorgrids,
ymajorgrids,
legend style={at={(0.04,0.1)}, anchor=south west, legend columns=1, legend cell align=left, align=left, draw=white}
]
\addplot [color=black, dotted, line width=1.5pt, mark=*, mark options={solid, fill=black, black}]
  table[row sep=crcr]{%
0	0.042456016377431\\
10	0.056408342144307\\
15	0.0562875168724661\\
20	0.0557490666006238\\
};
\addlegendentry{IR$_I$N-EE}

\addplot [color=blue, line width=1.5pt, mark=*, mark options={solid, fill=black, blue}]
  table[row sep=crcr]{%
0	0.038815360910628\\
10	0.0324876407210642\\
15	0.0168302478293593\\
20	0.0065489819755342\\
};
\addlegendentry{IR$_I$N}
\end{axis}
\end{tikzpicture}%
    \caption{The average fairness rate versus $P_c$ for SNR$=0$dB, $N_{BS}=1$, $N_{RIS}=15$, $L=2$, $K=2$, $M=2$.}
	\label{Fig-ee2}  
\end{figure}
In Fig. \ref{Fig-ee2}, we show the average fairness EE achieved by the NOMA-based IGS schemes proposed in Section III and in Section IV for $P_c=10$dB, $N_{BS}=1$, $N_{RIS}=15$, $L=2$, $K=2$, $M=2$. Indeed, we compare the performance of our proposed algorithms, which consider two different utility functions, from a fairness EE point of view. In other words, this figure shows how much EE is lost if the minimum rate of users  is maximized. The EE loss increases with the power budget. The reason is that if we consider a point-to-point system, the rate strictly increases with the transmission power. However, the EE is maximized by a specific  transmission power and is strictly decreasing for transmission powers higher than the optimum transmission power (see \cite[Fig. 2]{buzzi2016survey}).  When the  optimum transmission power for the MWEEM problem is lower than the power budget, the solution of the MWEEM problem does not change by increasing the power budget. That is why the fairness EE is almost constant for $P\geq 10$ dB. On the contrary, the transmission power of the solution of the MWRM problem is expected to be increased by the power budget, which results in a lower energy efficient operational point. Additionally, note that the MWEEM problem is equivalent to MWRM for a very large $P_c$ \cite{buzzi2016survey}. Thus, the algorithms for the MWRM problem perform close to the algorithms for the MWEEM problem when power budget is much lower than $P_c$.
 
\subsection{Summary of numerical results}
Our main findings in the numerical results can be summarized as in the following:
\begin{itemize}

\item RIS always improves the performance of the system; however, we have to properly optimize their elements to reap the full benefits of reconfigurable intelligent surfaces. In some cases, a RIS with random phases may even worsen performance, which shows the importance of optimizing  the RIS components.  

\item The combination of NOMA with appropriate user pairing plus IGS can highly improve the system performance from both spectral and energy efficiency perspectives when the BSs are overloaded, i.e., when $N_{BS}<2K$.  In this case, the interference level is high, which significantly degrades the system performance if interference is not properly managed.

\item Our results show that, even with a relatively low number of RIS components, RIS can provide a considerable performance gain. This suggests that RIS can be a promising technology from a practical point of view. 

\end{itemize}

\section{Conclusion}\label{sec-con}
In this paper, we proposed NOMA-based IGS schemes to improve the spectral and energy efficiency of multicell RIS-assisted BCs. We showed that RIS can significantly improve the system performance in both overloaded ($N_{BS}< 2K$) and underloaded ($N_{BS}\geq 2K$) regimes. In overloaded regimes,  RIS is unable to completely mitigate undesired consequences of interference, and we have to employ interference-management techniques to fully exploit the benefits of RIS. In other words, RIS cannot be used alone as an   interference-management technique when $N_{BS}< 2K$. However, in underloaded regimes, the interference level is low, and there are only minor benefits (if there exist any) by employing IGS/NOMA. We also showed that, even with a relatively low number of RIS components, RIS can considerably improve the system performance, which indicates that RIS can be efficiently implemented in practice. Moreover, we showed that RIS components should be carefully optimized. Otherwise, RIS may only provide minor benefits or even may decrease the system performance, especially in SISO systems.  

As a future work, it is interesting to investigate the performance of IGS and/or NOMA in the presence of imperfect CSI, employing a robust design similar to \cite{zhou2020framework}. Moreover, analyzing the performance gap between our proposed algorithms and the global optimal solution can be also another challenging direction for a future study.

\bibliographystyle{IEEEtran}
\bibliography{ref2}

\end{document}